\documentclass[fleqn]{statsoc}
\makeatletter\let\proof\@undefined\let\endproof\@undefined\makeatother
\usepackage[a4paper]{geometry}

\usepackage{amsfonts,url} \usepackage{amssymb} \usepackage{amsmath}
\usepackage{amstext,verbatim,graphicx,ifthen,multirow,amsthm}

\usepackage{algorithm}
\usepackage{booktabs}
\usepackage{makecell}
\usepackage{hyperref}
\usepackage{rotating}
\usepackage{lscape}

\usepackage[noend]{algpseudocode}
\usepackage[svgnames]{xcolor}

\usepackage{natbib}

\renewcommand{\thepage}{}
\renewcommand{\appendix}{\footnotesize\parindent 0cm\setcounter{equation}{0} 
\renewcommand{\theequation}{A.\arabic{equation}}
\setcounter{lemma}{0}\renewcommand{\thelemma}{A.\arabic{lemma}}}

\newcommand{\oy}{\overline{Y}}

\newtheorem{assumption}{Assumption}
\newtheorem{proposition}{Proposition}

\newtheorem{lemma}{Lemma}
\newtheorem{theorem}{Theorem}

\newtheorem{remark}{Remark}

\newcommand{\mme}{\mathbb{E}}

\newcommand{\pder}[3][]{\frac{\partial^{#1} #2}{\partial #3^{#1}}}

\usepackage{etoolbox}
\makeatletter
\patchcmd{\@makecaption}
  {\parbox}
  {\advance\@tempdima-\fontdimen2} 
  {}{}
\makeatother

\DeclareMathOperator{\var}{var}

\def\monthname{\ifcase\month\or
January\or February\or March\or April\or May\or June\or
July\or August\or September\or October\or November\or December\fi}

\renewcommand{\appendix}{\small\parindent 0cm\setcounter{equation}{0} 
\renewcommand{\theequation}{A.\arabic{equation}}
\setcounter{lemma}{0}\renewcommand{\thelemma}{A.\arabic{lemma}}
\setcounter{theorem}{0}\renewcommand{\thetheorem}{A.\arabic{theorem}}}
\title[Semiparametric Estimation of Treatment Effects in Randomized Experiments]{Semiparametric Estimation of Treatment Effects in Randomized Experiments}

\author{Susan Athey}
\address{Graduate School of Business, Stanford University, and NBER, USA.}
\email{athey@stanford.edu}

\author{Peter J. Bickel}
\address{Department of Statistics, University of California, Berkeley, USA.}
\email{bickel@stat.berkeley.edu}

\author{Aiyou Chen}
\address{Google LLC, Mountain View, USA.}
\email{aiyouchen@google.com}

\author{Guido W. Imbens}
\address{Professor of Economics, and Amman Mineral Faculty Fellow, Graduate School of Business and Department of Economics, Stanford University, and NBER, USA.}
\email{imbens@stanford.edu}

\author[Athey {\it et al.}] {Michael Pollmann}
\address{Department of Economics, Duke University, USA.}
\email{michael.pollmann@duke.edu}

\begin{document}

\begin{abstract}
We develop new semiparametric methods for estimating treatment effects. We focus on settings where the outcome distributions may be thick tailed, where treatment effects may be small, where sample sizes are large and where assignment is completely random. This setting is of particular interest in recent online experimentation. We propose using parametric models for the treatment effects, leading to semiparametric models for the outcome distributions. We derive the semiparametric efficiency bound for the treatment effects for this setting, and propose efficient estimators. In the leading case with constant quantile treatment effects one of the proposed efficient estimators has an interesting interpretation as a weighted average of quantile treatment effects, with the weights proportional to minus the second derivative of the log of the density of the potential outcomes. Our analysis also suggests an extension of Huber's model and trimmed mean to include asymmetry.
\end{abstract}


\baselineskip=15pt\newpage \setcounter{page}{1}\renewcommand{\thepage}{[
\arabic{page}]}\renewcommand{\theequation}{\arabic{section}.
\arabic{equation}}

\newpage

\section{Introduction}
\label{section:introduction}

Historically, randomized experiments were often carried out in medical and agricultural settings. In these settings,  sample sizes were often modest, typically on the order of hundreds or (more rarely) thousands of units. Outcomes commonly studied included  mortality or crop yield, and were characterized by relatively well-behaved distributions with thin tails. 
Standard analyses in those settings typically  involved  estimating the average effect of the treatment using the difference in average outcomes by treatment group, followed by constructing  confidence intervals using Normal distribution based
approximations. These methods originated in the 1920s, {\it e.g.},  \citet{neyman1923, fisher1937design}, but they continue to be the standard in modern applications. See  \citet{wu2011experiments} for a recent discussion.

More recently many experiments are conducted online  (see \cite{kohavi2020trustworthy} for an overview), leading to substantially different settings.  
 \citet[p.20]{gupta2019top} claim that  ``Together these organizations [Airbnb, Amazon,
Booking.com, Facebook, Google, LinkedIn, Lyft, Microsoft, Netflix, Twitter, Uber, and Yandex] tested more than one hundred thousand experimental treatments last year."
The settings for these online experiments are substantially different from those in biomedical and agricultural settings.
First, the experiments are often on a vastly different scale, with the number of units on the order of millions to tens of  millions. Second, the outcomes of interest, variables such as time spent by a consumer, sales per consumer or payments per service provider, are characterized by distributions with extremely thick tails. Third, the treatment effects are often extremely small relative to the standard deviation of the outcomes, even if their magnitude remains substantively important.
For example, \citet{lewis2015unfavorable} analyze challenges with statistical power in experiments designed to measure the effect of digital advertising on consumer expenditures. They discuss a hypothetical experiment where the average expenditure per potential customer is \$7 with a standard deviation of \$75, and where an average treatment effect of \$0.35 (0.005 of a standard deviation) would be substantial in the sense of being highly profitable for the company given the cost of advertising. In the Lewis and Rao  example, an experiment with power 0.8 for a treatment effect of \$0.35, and  a significance level for the two-sided test of means of 0.05, would require a sample size of 1.4 million customers.
As a result 
confidence intervals for the average treatment effect  are likely to include zero even if the true  effects were substantively important and samples are large.
Even if a confidence interval for the average treatment effect includes zero there may be evidence about the presence of causal effects of the treatment. Using Fisher exact p-value calculations \citep{fisher1937design} with well-chosen  statistics ({\it  e.g.}, the Hodges-Lehman difference in average ranks, \citet{rosenbaum1993hodges}), one may well be able to establish conclusively that treatment effects are present. However, the magnitude of the treatment effect, rather than its presence, is typically important for decision makers. 

This sets the stage for the problem we address in this paper. In the absence of additional information, there exists no estimator for the average treatment effect that is more efficient than the difference in means. To obtain more precise estimates, we either need to change the focus away from the average treatment effect, or we need to make additional assumptions.
One approach to changing the question, at least slightly, is to transform the outcome
({\it e.g.} taking logarithms or winsorizing) followed by a standard analysis estimating the average effect of the treatment on the transformed outcome.
In this paper, like 
\citet{taddy2016scalable, tripuraneni2021meta},  we choose a different approach, namely making additional assumptions on the joint distribution of the outcomes and treatment indicator.

The key conceptual contribution is that we postulate  a semi-parametric model for the outcome distributions by treatment group. The leading example of this semi-parametric model corresponds to restricting the quantile treatment effects to be identical across quantiles, thus assuming that the two conditional outcome distributions differ only by a shift.
We do not directly use parametric  models for the  outcome distributions by treatment group,  because specifying such a  model that well approximates the full outcome distribution is more challenging than postulating a model for the treatment effects. Unlike outcomes, treatment effects tend to be small and often have little variation. For this semiparametric set-up ({\it e.g.}, \citet{bickel1993efficient}, \citet{bickel2015mathematical}), we derive the influence function, the semiparametric efficiency bound, and we propose semiparametrically efficient estimators.

It turns out that the parametrization of the treatment effect can be very informative, potentially making the asymptotic variance for the corresponding semiparametric estimators substantially smaller than the  asymptotic variance for the difference-in-means estimator. 
For example, if the potential outcomes have Cauchy distributions, the  variance bound for the average treatment effect is infinite because the moments of the Cauchy distribution do not exist. However, under the constant additive treatment effect assumption (implying that the quantile treatment effects are identical),  the semiparametric variance bound for the treatment effect is finite.

In addition, even if this model for the treatment effect is misspecified, the estimand corresponding to   proposed estimators continue to have a causal interpretation, as a weighted average of quantile treatment effects, making it an easy-to-implement and attractive choice in practice.
 
The remainder of the paper is organized as follows.
First, in Section \ref{section:example} we consider the leading case where we assume  the two potential outcome distributions differ only by a shift, so that the quantile treatment effects are all identical. 
This is implied by, but does not require,  the assumption that the treatment effect is additive and constant. 
In Section \ref{estimation_model} we  consider the
case where we have more flexible parametric models linking the two conditional outcome distributions.
In Section \ref{section:simulations} we provide some simulation evidence regarding the finite sample properties of the proposed methods in controlled settings and provide real data illustrations.
Section \ref{section:conclusion} concludes.
A software implementation for R is available at \url{https://github.com/michaelpollmann/parTreat}.

\section{Constant Quantile Treatment Effects}\label{section:example}

In this section we focus on a special
case  with constant quantile treatment effects. After  setting up the problem formally we discuss robust estimation in the one-sample case to motivate a class of weighted quantile treatment effect estimators. We then discuss the formal semiparametric problem and show adaptivity of the proposed estimators.
Finally we consider partial adaptivity and robustness.

This case is closely related to the classical two-sample problem, as discussed in \citet{hodges1963estimates}, and to problems considered in the literature on robust descriptive statistics as in \citet{bickel1975descriptiveI, bickel1975descriptiveII}, \citet{doksum1974empirical}, \citet{doksum1976plotting}, and in particular
 \citet{jaeckel1971robust}, \citet{jaeckel1971some}.  Section \ref{section:example} can be interpreted   as an extension of Jaeckel's work, in a causal inference framework, to the two sample context in the setting of semiparametric theory. In the process of doing so, we generalize Huber's model \citep{huber1964} and the estimator based on trimmed means to include asymmetry, and present a simplified version of the results of \citet{chernoff1967asymptotic} (see also \citet{bickel1967some}, \citet{govindarajulu1967generalizations} and \citet{stigler1974linear} on linear combinations of order statistics). In particular, we exhibit efficient M (maximum-likelihood type) and L (linear combination of order statistics) estimates for outcome distributions that are known up to a shift. We then analyze  fully adaptive estimates of both types, as discussed in \citet{bickel1993efficient}, and partially adaptive estimates, in particular flexible trimmed means (\citet{jaeckel1971robust}).
In this setting the problem is closely related to the literature on robust estimation of locations ({\it e.g.,} \citet{bickel1975descriptiveI, bickel1975descriptiveII, hampel2011robust, huber2011robust, 
 bickel1976descriptive, bickel2012descriptiveIV}).

\subsection{Set Up}

We consider a set up with a randomized experiment with $n$ observations drawn randomly from a large population. With probability $p\in(0,1)$ a unit is assigned to the treatment group. Let  $n_1$ and $n_0=n-n_1$ denote the number of units assigned to the treatment and control group. 
Following \citep{neyman1923, rubin1974estimating, imbens2015causal},  let $Y_i(0)$ and $Y_i(1)$ denote the two potential outcomes for unit $i$, and let the treatment be denoted by $Z_i\in\{0,1\}$.
We assume that the treatment assignment for one unit does not affect the outcomes for any other unit. For all units in the sample we observe the pair $(Z_i,Y_i)$, where
 $Y_i\equiv Y_i(Z_i)$.
The cumulative distribution functions for the two potential outcomes are $F_0(y)$ and $F_1(y)$
with  inverses $F^{-1}_{0}(u)$ and   $F^{-1}_{1}(u)$, and means and variances $\mu_0$, $\mu_1$, $\sigma_0^2$ and $\sigma_1^2$.
Note that by the random assignment assumption the distribution of the potential outcome $Y_i(z)$ is identical to the conditional distribution of the realized outcome $Y_i$ conditional on $Z_i=z$: $F_z(y)\equiv \Pr(Y_i(z)\leq y)=\Pr(Y_i\leq y|Z_i=z)$.

We are interested in the average treatment effect in the population,
\begin{equation} \tau \equiv\mme[Y_i(1)-Y_i(0)].\end{equation}
The  natural estimator for  this average treatment effect is the difference in sample averages
\begin{equation} \hat\tau=\oy_1-\oy_0,\hskip1cm
{\rm where}\ \ 
\oy_1=\frac{1}{n_1}\sum_{i=1}^{n} Z_i Y_i,\hskip0.4cm  \oy_0=\frac{1}{n_0}\sum_{i=1}^{n} (1-Z_i)Y_i,\end{equation}
are the averages of the observed outcomes by treatment group. Under standard conditions $\frac{n_{1}}{n}\stackrel{P}{\rightarrow}p$,
 $\frac{n_{0}}{n}\stackrel{P}{\rightarrow}(1-p)$, and
\begin{align}\label{vvar}
\sqrt{n}(\hat{\tau}-\tau) & \stackrel{d}{\Rightarrow} \mathcal{N}\left(0,\frac{\sigma^2_0}{1-p}+\frac{\sigma^2_1}{p}\right).
\end{align}

The concern  is that  this conventional estimator $\hat\tau$ may be imprecise.
In particular in settings where the outcome distribution is thick tailed, sometimes extremely so, confidence intervals may be wide.
We address this issue in this paper by 
imposing some restrictions on the two potential outcome distributions. Following  the semiparametric literature \citep{bickel1975descriptiveI, bickel1975descriptiveII, bickel1976descriptive, bickel2012descriptiveIV}  we exploit these restrictions to develop new estimators.
 
 \subsection{Weighted Average Quantile Treatment Effects}
 It is useful to start with quantile treatment effects  \citep{lehmann1975nonparametrics}, which play an important role in our setup. For quantile $u\in(0,1)$, define
\begin{equation}\label{response}
{\Delta}(u) \equiv F_{1}^{-1}(u)-F_{0}^{-1}(u),\ \ 0\leq u\leq 1.
\end{equation}
These quantile treatment effects are closely related to what
\citet{doksum1974empirical} and \citet{doksum1976plotting} 
  label the {\it response} function:
$R(y) \equiv
F^{-1}_{1}\left( F_{0}(y)\right)-y=\Delta(F_0(y))
.$
The natural estimate for the quantile treatment effect 
is the empirical plug-in, 
$
\hat{{\Delta}}(u)\equiv\hat{F}_{1}^{-1}(u)-\hat{F}_{0}^{-1}(u)
$
where $\hat{F}_{1}^{-1}(u)$ equals $Y_{([n_1 u])}^{(1)}$, defined as the $[n_1u]^{\rm th}$ order
statistic of $Y_{i}|Z_{i}=1$, $i=1,\dots,n_{1}$, where $n_{1}=\sum_{i=1}^{n}Z_{i}$
and similarly for $\hat{F}_{0}^{-1}(u)$.

A natural class of parameters summarizing the difference
between the $Y_i(1)$ and $Y_i(0)$ distributions consists of weighted averages of the quantile treatment effects:
\begin{align*}
\tau(F_{0},F_{1};W) & \equiv\int_{0}^{1}{\Delta}(u)dW(u)
\end{align*}
where the weights integrate to one, $W(0)=0$, $W(1)=1$. 
Different choices for the weight function correspond to different estimands.
The constant weight case, $W'(u)\equiv 1$, corresponds to
the population average treatment effect $\tau=\mathbb{E}[Y_i(1)-Y_i(0)]$.
The median corresponds to the case where  $W(\cdot)$ puts all its mass at
$1/2$. 
We thus allow $W(\cdot)$ to permit point masses.

For a given weight function $W(\cdot)$ we can estimate the parameter $\tau(F_{0},F_{1};W)$ 
using a {\it  weighted average quantile} estimator: 
\begin{equation}
\label{waqe}
\hat{\tau}_{W}  \equiv\tau(\hat F_{0},\hat F_{1};W)=\int_{0}^{1}\Bigl(\hat{F}_{1}^{-1}(u)-\hat{F}_{0}^{-1}(u)\Bigr)
dW(u)\end{equation}
\[\hskip1cm=\frac{1}{n_{1}}\sum_{i=1}^{n_{1}}w_{i}^{(1)}Y_{(i)}^{(1)}-\frac{1}{n_{0}}\sum_{i=1}^{n_{0}}w_{i}^{(0)}Y_{(i)}^{(0)},
\]
where 
\begin{align*}
w_{i}^{(z)} & \equiv W\left(\frac{i}{n_{z}}\right)-W\left(\frac{i-1}{n_{z}}\right),
\end{align*}
and $Y_{(i)}^{(z)}$ again are  the order statistics in treatment group $z$.

\subsection{Efficient estimation of Weighted Average Quantile Treatment Effects Using Influence Functions}

To understand the properties of the weighted average quantile estimator $\hat\tau_W$, we begin by considering the nonparametric model for a single sample,
with cumulative distribution $F(\cdot)$ where the interest is in the weighted
quantile $\int_{-\infty}^{\infty}F^{-1}(u)dW(u)$ for a given weight funtion $W(\cdot)$. For this one-sample
case, \cite{jaeckel1971robust,jaeckel1971some}, building on \cite{chernoff1967asymptotic},
shows that under simple conditions on $W$ and $F$, for a sample
size of $n$, 
\begin{align}
\label{een}
\int_{0}^{1}\Bigl(\hat{F}^{-1}(u)-F^{-1}(u)\Bigr)dW(u)=\int_{-\infty}^{\infty}\psi(x,F,W)d(\hat{F}(x)-F(x))+o_{P}(n^{-1/2}),
\end{align}
where the  {\it influence function} $\psi$ is related to the weight function
$W$ by
\begin{align}
\psi(x,F,W) & = -\int_{x}^{\infty}\frac{1}{f(y)}dW(F(y)) + \int_{-\infty}^{\infty}\frac{F(y)}{f(y)}dW(F(y))\label{eq:2.5}.
\end{align}
The last term ensures that $\int_{-\infty}^{\infty}\psi(x,F,W)dF(x)=0$.

Note that if the derivatives $\psi'(\cdot)$ and $W'(\cdot)=w(\cdot)$ exist, by \eqref{eq:2.5},
\begin{align*}
\psi(x,F,W) & =-\int_{x}^{\infty}w(F(y))dy+\int_{-\infty}^{\infty}F(y)w(F(y))dy
\end{align*}
so that $\psi'(x,F,W)=w(F(x))$. 
Note that for the
median, \eqref{eq:2.5} yields, $\psi(x,F,W)=\frac{sign(x-F^{-1}(\frac{1}{2}))}{2f(F^{-1}(\frac{1}{2}))}$.
Our formula \eqref{eq:2.5} is slightly more general than Jaeckel's, and in the appendix we establish sufficient conditions on the cumulative distribution function $F(\cdot)$ and the weight function $W(\cdot)$ for our version of his result to hold.

Expression (\ref{een}) in turn implies that 
\begin{align*}
\int_{0}^{1}\Bigl(\hat{F}^{-1}(u)-F^{-1}(u)\Bigr)dW(u) & \stackrel{d}{\Rightarrow}\mathcal{N}(0,\sigma^{2}(F,W)).
\end{align*}
where the variance equals the expectation of the square of the influence function:
\[
\sigma^{2}(F,\omega)=\int_{-\infty}^{\infty}\psi(x,F,W)^{2}dF(x).
\]
The results in 
\cite{jaeckel1971robust,jaeckel1971some} for the one-sample case
extend in the following way to the two-sample setting that is our primary focus.
If $\tau(F_{0},F_{1},W)$ is estimated  by $\hat{\tau}_W$ in (\ref{waqe}),
then, under regularity conditions
given in the appendix (Theorem A.1),
\begin{align}
\hat{\tau}_W& =\tau(F_{0},F_{1},W)+\frac{1}{n}\sum_{i=1}^{n}\left(Z_{i}\frac{\psi(Y_{i},F_{1},W)}{p}-(1-Z_{i})\frac{\psi(Y_{i},F_{0},W)}{1-p}\right)+o_{P}(n^{-1/2})\label{eq:2.6}
\end{align}
where $\psi(x,F,W)$ is given by \eqref{eq:2.5}.

\subsection{Constant Quantile Treatment Effects}

Now let us return to the primary focus of this section, the estimation of the average treatment effect under the constant quantile treatment effect assumption. Our key assumption in this section is that
the quantile treatment effects are all equal:
\begin{align}
\label{equal_quantile}
{\tau}(u) & =\tau,\hskip1cm \forall  \ u, 
\end{align}
and thus, for any weight functions $W(\cdot)$,
\begin{align}
\tau(F_{0},F_{1},W) & =\tau.
\end{align}
Later, in Section 3, we generalize this to allow for a more general parametric function linking the quantile treatment effects.
One way to motivate the constant quantile treatment effect assumption is to 
 assume that the unit-level treatment effects are all constant, $Y_i(1)-Y_i(0)=\tau$ for all units $i=1,\ldots,n$. This implies, but is not implied by, the assumption that all the quantile treatment effects are identical.
The assumption of constant unit-level treatment effects is very strong, implying rank-invariance, which is in fact stronger than what we need.

In this section, for expository reasons we further assume that 
we know the control outcome distribution $F_{0}(\cdot)$ up to a shift. That is, $F_{0}(x)=F(x-\eta)$ where $F(\cdot)$ (with derivative $f$) is known and  
$\eta$ unknown. Because of the constant quantile treatment effect assumption the treated potential outcome distribution is also known up to a shift, $F_{1}(x)=F(x-\eta-\tau)$.
Assuming that $F_0(\cdot)$ is known up to a shift is unrealistic in practice, and we remove this assumption below in Section 2.5, but it allows us to focus in this section on some key insights.

For this fully parametric model (with unknown parameters $\eta$ and
$\tau$), if the Fisher information $I(f)=\int\left(\frac{f'}{f}\right)^{2}(x)f(x)dx=\int \left(-\frac{f'}{f}\right)'(x)f(x)dx$ satisfies $0<I(f)<\infty$, the maximum likelihood estimator  of $\tau$, suitably
regularized ({\it e.g.,}  \citet{lecam-yang-1988}), has
influence function,
\begin{align}
\psi_{f,\eta}(Z,Y; \tau) & =-\frac{1}{I(f)}\cdot\left(\frac{Z}{p}\cdot\frac{f'}{f}(Y-\eta-\tau)-\frac{1-Z}{1-p}\cdot\frac{f'}{f}(Y-\eta)\right)\label{eq:psi_f(Z,Y)}.
\end{align}

There is an interesting alternative efficient estimator in this known $f(\cdot)$ case.
Suppose $f'/f$ is absolutely continuous. Then the weight function
\begin{equation}
w_{f}(F(x))  \equiv \frac{1}{I(f)}\left(-\frac{f'}{f}\right)'(x)  \text{ or } w_f(u) = \frac{1}{I(f)}\left(-\frac{f'}{f}\right)'(F^{-1}(u)) \label{eq:w_f(F(x))}
\end{equation}
provides an efficient $L$ estimate when substituted appropriately
in \eqref{een},
leading to \begin{align} \hat{\tau}(F_{0},F_{1},W) & =\int_{0}^{1}(\hat{F}_{1}^{-1}(u)-\hat{F}_{0}^{-1}(u)) w_f(u) du.
\label{eq:tau-hat-F0-F1-ww}
\end{align}
It is interesting to inspect the form of the weights $w_f(u)$. These weights are proportional to minus the second derivative of the logarithm of the density function. In other words, we can approximate the efficient estimator by first estimating a large number of quantile treatment effects. Under the model these quantile treatment effects are all identical. To efficiently estimate that common treatment effect we can simply use a weighted average of the estimated quantile treatment effects.  It turns out the optimal weights simplify to minus the second derivative of the logarithm of the density. For the Normal distribution, that means the weights are constant. For the double exponential distribution the weights put point mass at the median. For the Cauchy distribution the weights are proportional to 
$-\cos(2\pi u) \sin(\pi u)^{2}$. Interestingly these weights are negative for some quantiles. One can of course see this by inspecting the estimated weights. If one is concerned by the negative weights one can also modify them by restricting them to be nonnegative. Finally, note that implicitly the influence function estimator also has the negative weights in such cases because the two estimators are first order equivalent.

A final comment connects this to common methods for dealing with thick tailed distributions. In practice many researchers use winsorizing to deal with these problems. This can be interpreted as using a weighted average quantile estimator
with a particular set of weights. Specifically, with winsorizing at the $q$ and $1-q$ quantiles, the implicit weights  are constant on the interval $(q,1-q)$, and then put additional point mass $q$ on the $q$th and $(1-q)$th quantiles. As discussed in \citet{bickel1965some}, the asymptotic properties of the winsorizing estimator depend delicately on the density at the winsorizing quantiles. In our simulations this estimator does not perform particularly well. Like other settings, there is tension here between having an interpretable target that may not be precisely estimable ({\it e.g.,} the average effect of the treatment), versus a precisely estimable estimand whose interpretation is more complex ({\it e.g.} the weighted average quantile effect). This tension arises also in other settings. An example is the estimation of average treatment effects under unconfoundedness where weighting by the confounders may affect both the interpretation  of the estimand and the precision with which we can estimate it \citep{crump2009dealing, li2018balancing}.
Another setting is that discussed in \citet{vansteelandt2022assumption}.
The use of quantile methods for estimating treatment effects in thick-tailed settings has been studied in \citet{firpo2007efficient, firpo2009unconditional}, but unlike in those papers, our focus is on the overall treatment effect, rather than the effect at specific quantiles.

\subsection{Fully adaptive estimation}

As stated earlier, in practice we do not know the density $f(\cdot)$ up to location. However, in this case with constant quantile treatment effects this knowledge does not matter up to first order.
Because of the 
orthogonality of the tangent space with respect to $f$,
it follows from semiparametric theory \citep{bickel1993efficient} that even if the density $f$
is unknown, substituting a suitable estimate of $f$ (and $\eta$) in \eqref{eq:psi_f(Z,Y)}
or (\ref{eq:tau-hat-F0-F1-ww}), will yield estimators with influence functions
given by \eqref{eq:psi_f(Z,Y)}, or equivalently by
\begin{align}\label{eq:psi-additive}
\psi_{f_0}(Z,Y; \tau) & =-\frac{1}{I(f_0)}\cdot\left(\frac{Z}{p}\cdot\frac{f_0'}{f_0}(Y-\tau)-\frac{1-Z}{1-p}\cdot\frac{f_0'}{f_0}(Y)\right),
\end{align}
where $f_0(\cdot) \equiv F_0'(\cdot) \equiv f(\cdot - \eta)$.

For our proposed estimator we split the data randomly into two parts, with the two subsamples denoted by $A$, corresponding to $\{(Z_{i},Y_{i}): 1\leq i \leq \frac{n}{2}\}$, and $B$, corresponding to $\{(Z_{i},Y_{i}): \frac{n}{2} < i\leq  n\}$. 
The $M$ estimate using the estimated $\hat f(\cdot)$  is of the form,
\begin{align}
\hat{\tau}^{if} & \equiv \frac{\tilde{\tau}_{(A)} + \tilde{\tau}_{(B)}}{2} +\frac{1}{n}\left\{\sum_{i=1}^{n/2}\psi_{\hat{f}_{0(B)}}(Z_{i},Y_{i};\tilde{\tau}_{(B)})+\sum_{i=1+n/2}^{n}\psi_{\hat{f}_{0(A)}}(Z_{i},Y_{i}; \tilde{\tau}_{(A)})\right\} \label{eq:tau-if}
\end{align}
where $\hat{f}_{0(A)}$ is an estimate of $f_0$ using
$\{(Z_{i},Y_{i}): 1\leq i \leq \frac{n}{2}\}$, and $\hat{f}_{0(B)}$ using
$\{(Z_{i},Y_{i}): \frac{n}{2} < i\leq  n\}$, a one step estimate
using the sample splitting technique \citep{klaassen1987consistent}. $\tilde\tau_{(A)}$ and $\tilde\tau_{(B)}$ are initial $\sqrt n$ consistent estimates based on the two subsamples, for example based on the difference in medians or other quantiles.
Algorithm \ref{algo:eif} shows the key steps; additional details are given in the supplementary materials.

\begin{algorithm}
\caption{Influence Function Based Estimator $\hat\tau^{if}$}\label{algo:eif}
\begin{algorithmic}[1]
\\ {$\rhd$ Input: }
\\ \hskip0.6cm {$n_{1}$ treated observations $Y_1^{1},\ldots,Y_{n_1}^{(1)}$}
\\ \hskip0.6cm {$n_{0}$ control observations $Y_1^{0},\ldots,Y_{n_0}^{(0)}$}
\\
\\ {$\rhd$ Randomly split sample into halves $A$ and $B$:}
\\ \hskip0.6cm $n_{1(A)} = \lceil n_1/2\rceil$, $n_{1(B)} = \lfloor n_1/2\rfloor$, $n_{0(A)} = \lceil n_0/2\rceil$, $n_{0(B)} = \lfloor n_0/2\rfloor$
\\ \hskip0.6cm denote treated in halves $A$ and $B$ by $Y^{(1,A)}_1,\ldots,Y^{(1,A)}_{n_{1(A)}}$ and $Y^{(1,B)}_1,\ldots,Y^{(1,B)}_{n_{1(B)}}$
\\ \hskip0.6cm denote control in halves $A$ and $B$ by $Y^{(0,A)}_1,\ldots,Y^{(0,A)}_{n_{0(A)}}$ and   $Y^{(0,B)}_1,\ldots,Y^{(0,B)}_{n_{0(B)}}$
\\
\\ {$\rhd$ Calculate a preliminary consistent estimator:}
\\ \hskip0.6cm $\tilde{\tau}_{(B)} = \mathrm{median}(Y^{(1,B)}_{i}) - \mathrm{median}(Y^{(0,B)}_{i})$
\\
\\ {$\rhd$ Estimate density and its derivatives:}
\\ \hskip0.6cm $\hat{f}_{0(B)}(\cdot)$, $\hat{f}_{0(B)}'(\cdot)$, $\hat{f}_{0(B)}''(\cdot) \gets$  estimated using data $Y^{(0,B)}_1,\ldots,Y^{(0,B)}_{n_{0(B)}}$
\\
\\ {$\rhd$ Estimate the Fisher information $I$:}
\\ \hskip0.6cm $\hat{I}_{(B)} \gets - \frac{1}{n_{0(B)}} \sum_{i=1}^{n_{0(B)}} \frac{\hat{f}_{0(B)}(Y^{(0,B)}_{i})\hat{f}''_{0(B)}(Y^{(0,B)}_{i}) - \hat{f}'_{0(B)}(Y^{(0,B)}_{i})^2}{\hat{f}_{0(B)}(Y^{(0,B)}_{i})^2}$
\\
\\ {$\rhd$ Estimate the effects:}
\\ \hskip0.6cm $\hat{\tau}_{(A)} \gets \tilde{\tau}_{(B)} + \frac{1}{n_{1(A)}} \sum_{i=1}^{n_{1(A)}} \frac{-1}{p \hat{I}_{(B)}} \frac{\hat{f}'_{0(B)}(Y^{(1,A)}_{i} - \tilde{\tau}_{(B)})}{\hat{f}_{0(B)}(Y^{(1,A)}_{i} - \tilde{\tau}_{(B)})} - \frac{1}{n_{0(A)}} \sum_{i=1}^{n_{0(A)}} \frac{-1}{(1-p) \hat{I}_{(B)}} \frac{\hat{f}'_{0(B)}(Y^{(0,A)}_{i})}{\hat{f}_{0(B)}(Y^{(0,A)}_{i})}$
\\
\\ {$\rhd$ Repeat lines 10 through 20 reversing $A$ and $B$, then average:}
\\ \hskip0.6cm $\hat\tau^{if} \gets (\hat\tau_{(A)} + \hat\tau_{(B)})/2$
\end{algorithmic}
\end{algorithm}

We can also construct an $L$ estimate based on an average of the quantile differences. This estimator is obtained by 
first estimating  $F_0(\cdot)$, $f_0(\cdot)$, and $f_0'(\cdot)$, substituting that for $F(\cdot)$, $f(\cdot)$, and $f'(\cdot)$ into $w_f(u)$ in equation (\ref{eq:w_f(F(x))}), followed by using this estimated set of weights in (\ref{eq:tau-hat-F0-F1-ww}), leading to
\begin{align} \hat{\tau}^{waq} & =\int_{0}^{1}(\hat{F}_{1}^{-1}(u)-\hat{F}_{0}^{-1}(u)) \hat w_f(u) du.
\label{eq:tau-hat-F0-F1-ww_feasible}
\end{align}
Formally we would use the same sample splitting as above. Details
are in Algorithm \ref{algo:waq} and the supplementary materials.

\begin{algorithm}
\caption{Weighted Average Quantile Estimator $\hat\tau^{waq}$}\label{algo:waq}
\begin{algorithmic}[1]
\\ {$\rhd$ Input: }
\\ \hskip0.6cm {$n_{1}$ treated observations $Y_1^{1},\ldots,Y_{n_1}^{(1)}$}
\\ \hskip0.6cm {$n_{0}$ control observations $Y_1^{0},\ldots,Y_{n_0}^{(0)}$}
\\
\\ {$\rhd$ Randomly split sample into halves $A$ and $B$:}
\\ \hskip0.6cm $n_{1(A)} = \lceil n_1/2\rceil$, $n_{1(B)} = \lfloor n_1/2\rfloor$, $n_{0(A)} = \lceil n_0/2\rceil$, $n_{0(B)} = \lfloor n_0/2\rfloor$
\\ \hskip0.6cm denote treated in halves $A$ and $B$ by $Y^{(1,A)}_1,\ldots,Y^{(1,A)}_{n_{1(A)}}$ and $Y^{(1,B)}_1,\ldots,Y^{(1,B)}_{n_{1(B)}}$
\\ \hskip0.6cm denote control in halves $A$ and $B$ by $Y^{(0,A)}_1,\ldots,Y^{(0,A)}_{n_{0(A)}}$ and   $Y^{(0,B)}_1,\ldots,Y^{(0,B)}_{n_{0(B)}}$
\\
\\ {$\rhd$ Estimate density and its derivatives:}
\\ \hskip0.6cm $\hat{f}_{0(B)}(\cdot)$, $\hat{f}_{0(B)}'(\cdot)$, $\hat{f}_{0(B)}''(\cdot) \gets$  estimated using data $Y^{(0,B)}_1,\ldots,Y^{(0,B)}_{n_{0(B)}}$
\\

\\ {$\rhd$ Order and pair observations:}
\\ \hskip0.6cm $n_{(A)} = \max(n_{1(A)},n_{0(A)})$
\\ \hskip0.6cm duplicate treated or control observations as needed such that there are $n_{(A)}$ of both, evenly across the distribution, and order them (analogously for the $B$ split):
\\ \hskip1.2cm $Y^{(0,A)}_{(1)} \leq Y^{(0,A)}_{(2)} \leq \dots \leq Y^{(0,A)}_{(n_{(A)})}$; $Y^{(1,A)}_{(1)} \leq Y^{(1,A)}_{(2)} \leq \dots \leq Y^{(1,A)}_{(n_{(A)})}$
\\ \hskip1.2cm $Y^{(0,B)}_{(1)} \leq Y^{(0,B)}_{(2)} \leq \dots \leq Y^{(0,B)}_{(n_{(B)})}$; $Y^{(1,A)}_{(1)} \leq Y^{(1,B)}_{(2)} \leq \dots \leq Y^{(1,B)}_{(n_{(B)})}$
\\
\\ {$\rhd$ Estimate the weighted average quantile effect:}
\\ \hskip0.6cm weights: $w^{(B)}_{(i)} \gets - \frac{\hat{f}_{0(B)}(Y_{(i)}^{(0,B)})\hat{f}''_{0(B)}(Y_{(i)}^{(0,B)}) - \hat{f}'_{0(B)}(Y_{(i)}^{(0,B)})^2}{\hat{f}_{0(B)}(Y_{(i)}^{(0,B)})^2}$
\\ \hskip0.6cm $\hat{\tau}_{(A)} \gets \sum_{i=1}^{n_{(A)}} w^{(B)}_{(i)} (Y^{(1,A)}_{(i)} - Y^{(0,A)}_{(i)}) / \sum_{i=1}^{n_{(A)}} w^{(B)}_{(i)}$
\\
\\ {$\rhd$ Repeat lines 10 through 21 reversing $A$ and $B$, then average:}
\\ \hskip0.6cm $\hat\tau^{waq} \gets (\hat\tau_{(A)} + \hat\tau_{(B)})/2$
\end{algorithmic}
\end{algorithm}

Formally we have for the unknown $f(\cdot)$ case:
\begin{theorem}
\label{thm:efficiency}For all $f$ such that $f'$ exists and $0<I(f)<\infty$:
\begin{enumerate}
\item[$(i)$] There exist a $\sqrt{n}$-consistent estimator $\hat{\tau}$.
\item[$(ii)$] Under mild conditions (see \eqref{eq:fhat-score-convergence} and \eqref{eq:fhat-info-convergence} in the appendix), we can construct an $M$ estimate $\hat{\tau}^{if}$ such that 
\begin{align}
\sqrt{n}(\hat{\tau}^{if}-\tau) & \stackrel{d}{\Rightarrow}\mathcal{N}\left(0,\frac{1}{p(1-p)I(f)}\right).\label{eq:(A)}
\end{align}
\item[$(iii)$] Under mild conditions (see Lemma \ref{lem:linear-order-statistics} in the appendix), we can construct an $L$ estimate $\hat{\tau}^{waq}$ (weighted average quantile)
such that $\hat{\tau}^{waq}$ also satisfies \eqref{eq:(A)}.
\end{enumerate}
\end{theorem}

The main insight is that the asymptotic variance for the proposed estimator is the same as the variance for the maximum likelihood estimator in the case where $F_0$ is known up to a shift.
Our conditions are not optimal (see \citet{stone1975adaptive} for minimal ones in
the one sample case).
The $\sqrt n$-consistent estimator for $\tau$ can be based on any quantile treatment effect estimator. 

Thus, both the estimators $\hat{\tau}^{waq}$ and $\hat{\tau}^{if}$
are adaptive to models in which the distribution of control and treated
potential outcomes is known up to location. Theorem \ref{thm:efficiency}
implies that the influence function based estimator is as efficient
as the maximum likelihood estimator based on the true distribution
function. For instance, if potential outcomes are normally distributed,
the maximum likelihood estimator is the difference in means, and the
influence function estimator has the same limiting distribution. If,
however, the potential outcomes follow a double exponential distribution,
the difference 
in medians is the efficient estimator. Under this distribution,
the influence function based estimator adapts and has the same limiting
distribution as the difference in medians. For the Cauchy distribution the optimal weights are more complicated,
\(w_f(u) \propto - \cos(2 \pi u) \sin(\pi u)^2\),
but the influence function based estimator has the same limiting distribution,
without requiring a priori knowledge about the distribution.
This influence function based estimator $\hat{\tau}^{if}$ is a special
case of the estimator developed by \citet{cuzick1992efficient,cuzick1992semiparametric}
for the partial linear regression model setting.

Although these estimators are efficient under the constant additive
treatment model, as we shall see in Section 3, if the the constant quantile treatment effect assumption is violated,
estimates of the types derived from $f$ known continue to estimate
at rate $n^{-1/2}$, meaningful measures of the treatment
effect as discussed in Section 2.1. This is unfortunately not the
case for $\hat{\tau}^{if}$ and $\hat{\tau}^{waq}$ because estimation
of $f'$ and $f''$ introduces components of variance of order larger
than $n^{-1/2}$. However, there is a partial remedy, that we discuss next.

\subsection{Partial Adaptation}\label{section:partial}

An interesting alternative to fully adaptative estimation in the closely related one-sample symmetric case was
studied by \citet{jaeckel1971some}.  See also 
\citet{yu2017robust}. The estimator proposed by \citet{jaeckel1971some} is
\begin{align*}
\hat{\mu}_{\alpha} & \equiv\frac{1}{1-2\alpha}\int_{\alpha}^{1-\alpha}\hat{F}^{-1}(u)du
\end{align*}
for a sample from $F(x-\mu)$ with $f$ symmetric. We can interpret this as restricting the class of weight functions to one indexed by a scalar parameter $\alpha$:
\[
w_{\alpha}(u)  =
\left\{\begin{array}{ll}
\frac{1}{1-2\alpha}\hskip1cm& {\rm if}\ \alpha\leq u\leq1-\alpha\\
  0& \text{ otherwise}.
  \end{array}\right.
\]
This weight function
yields the $\alpha$-trimmed mean. We can then choose the value of $\alpha$ 
that minimizes the asymptotic
variance of $\hat{\mu}_{\alpha}$. This asymptotic variance is equal to 
\begin{align*}
\sigma_{\alpha}^{2}\equiv\frac{1}{n(1-2\alpha)^{2}} & (\mathbb{E}(X-\overline{\mu})^{2}1(F^{-1}(\alpha)\leq X\leq F^{-1}(1-\alpha))\\
 & +\alpha\cdot(F^{-1}(1-\alpha)-\overline{\mu})^{2}+\alpha(F^{-1}(\alpha)-\overline{\mu})^{2}),
\end{align*}
with 
\begin{align*}
    \overline{\mu} = \int_{\alpha}^{1-\alpha} F^{-1}(u)du + \alpha (F^{-1}(\alpha) + F^{-1}(1-\alpha)),
\end{align*}
which can be estimated by replacing $F$ with the empirical distribution, denoted as $\hat{\sigma}^2_{\alpha}$.
Let $\hat\alpha=\arg\min_\alpha\hat\sigma^2_\alpha$, and let $\hat{\mu}_{\hat\alpha}$
be the corresponding estimator for $\mu$.
\citet{jaeckel1971some} shows that $\hat{\mu}_{\hat\alpha}$ is adaptive for estimating $\mu$ over a Huber family of densities. In the Huber family  $(X-\mu)/\sigma$ has density $f$ for varying $\mu$, $\sigma>0$:
\[
\log f(x)= \left\{
\begin{array}{ll}
-\frac{x^{2}}{2}-c(k),\hskip1cm & {\rm if}\ |x|\leq k\\ \\
-\frac{k|x|}{2}-c(k), & {\rm if}\ |x|>k.
\end{array}\right.
\]
where $c(k)$ makes $\int f(x)dx=1$, and
$k=-F^{-1}(\alpha)$. Adaptivity here means that
using the trimming proportion optimizing the variance estimate,
in fact yields an estimate which is efficient for the member of the Huber family generating the data. He optimizes $0<\alpha_0 \leq \alpha \leq \alpha_1 <\frac{1}{2}$.

Because this family includes among others the 
 Gaussian ($k\rightarrow \infty$) and  double
exponential ($k=0$), this family
is very flexible. For more properties, see \cite{huber2011robust}.

In the two-sample case, it is reasonable to consider asymmetric weight
functions leading to the natural generalization, 
\begin{align}
\label{taualphabeta}
\hat{\tau}_{\alpha,\beta} & =\frac{1}{1-(\alpha+\beta)}\int_{\alpha}^{1-\beta}(\hat{F}_{1}^{-1}(u)-\hat{F}_{0}^{-1}(u))du.
\end{align}
This estimator is partially adaptive, in a similar way to the symmetric trimmed mean in the one sample problem.
In the online appendix, we extend Jaeckel's result on partial adaptation for our two sample problem to a generalization of the  Huber family whose members are symmetric iff $k_{1}=k_{2}$, defined by
$(X-\mu)/\sigma \sim f$:
\begin{align}
\log f(x)=\begin{cases}
-\frac{x^{2}}{2}-c(k_{1},k_{2}), & \text{ if }-k_{1}\leq x\leq k_{2},\\
\frac{k_{1}x}{2}-c(k_{1},k_{2}), & \text{ if }x<-k_{1},\\
-\frac{k_{2}x}{2}-c(k_{1},k_{2}), & \text{ if }x>k_{2},
\end{cases} \label{eq:extended-huber}
\end{align}
where $c(k_{1},k_{2}) \equiv \log \Bigl(2\left(\frac{e^{-k_{1}/2}}{k_{1}} + \frac{e^{-k_{2}/2}}{k_{2}}\right) + \sqrt{2\pi}(\Phi(k_{2}) - \Phi(-k_{1}))\Bigr)$ and $\Phi$ is the CDF of $\mathcal{N}(0, 1)$.
See Figure~\ref{fig:generalized-Huber} for illustration.
This family can be equivalently parametrized by $F(-k_{1})$ and $F(k_{2})$. $f(\cdot)$ is symmetric if $k_1=k_2$. 

\begin{figure}
    \centering
    \includegraphics{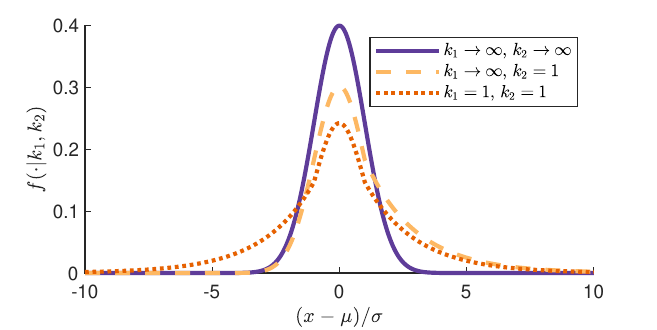}
    \caption{Example members of the generalized Huber family of distributions.}
    \label{fig:generalized-Huber}
\end{figure}

In the supplementary materials we also discuss how inference can proceed in this setting.

\section{The General Parametric Treatment Effect Case}
\label{estimation_model}

In some settings, the assumption of an additive model may be too restrictive.
In this section, we develop estimators given a general parametric model for this difference.

The starting point is a model governing the relation between the two potential outcomes:
\begin{assumption}[Parametric Model Quantile Treatment Effects]\label{as:par-outcomes}
The  potential outcome distributions satisfy
\begin{equation*}
    F_1(h(y,\theta)) = F_0(y).
\end{equation*}
\end{assumption}
The constant quantile treatment effect case is  a special case of this with $h(y,\theta)=y+\theta.$ Another important special case is the proportional treatment effect case, $h(y,\theta)=\theta y$. For the general case the weighted average quantile estimator does not directly generalize, so we focus on the influence-function-based estimator. For the general case the influence function is more complex.

This approach of modelling treatment effects has connections to the literature on structural nested models, which also imposes modeling restrictions on treatment effects, although for different reasons, largely based on the challenges in dynamic settings. See \citet{robins1986new} for an early paper, and \citet{vansteelandt2014structural} for a review.

As before we initially assume $F_0$ known.
In terms of the quantile treatment effects $\tau(u)$  Assumption \ref{as:par-outcomes} implies the restriction
\[ \tau(u)=F_1^{-1}(u)-F^{-1}_0(u)=h(F_0^{-1}(u),\theta)-F^{-1}_0(u)
.\]
Given Assumption \ref{as:par-outcomes}, the population average treatment effect can be characterized as
\begin{equation*}
\tau^{\rm pop}=\int_0^1 \bigl( h(F_0^{-1}(u),\theta)-F_0^{-1}(u)\bigr) du
\end{equation*}
In practice, however, estimating $\tau^{\rm pop}$ may still be subject to substantial sampling variance, even if $h(\cdot)$ is known. For example, suppose that $h(y,\theta)=\theta y$, so that the treatment effect is proportional. The average treatment effect is then $\theta E[Y_i|Z_i=0]$. Even if $\theta$ is known, estimating the population mean $E[Y_i|Z_i=0]$ could lead to a large standard error.
As an alternative, we therefore focus on a different estimand. Specifically, we suggest to estimate the in-sample, as opposed to population, average treatment effect. This is still a well-defined average causal effect that is useful for decision makers. It is in the spirit of the typical analysis of randomized experiments based on convenience samples where the focus is on the average effect for the particular sample.  A key insight is that estimators of this object can  have a much lower variance.
We define the in-sample average treatment effect as
\begin{equation}\label{tau_is}
\begin{aligned}
\tau^{\rm is} = \frac{1}{N} \sum_{i=1}^{N}\Bigl(Y_i(1) - Y_i(0)\Bigr) 
 = \frac{1}{N}\sum_{i=1}^{N} \Bigl\{ Z_i (Y_i - h^{-1}(Y_i,\theta)) + (1-Z_i)(h(Y_i,\theta) - Y_i)\Bigr\}.
\end{aligned}
\end{equation}
When \(h\) is not just an additive function, \(\tau^{\rm is}\) is sample-dependent, and thus stochastic.
In particular when the variance of \(Y_i\) is large because of thick tails for the potential outcome distributions, the variance of \(\tau^{\rm is}\) over repeated samples can be large, too.
To give some intuition for this, suppose that $\theta$ is known. Then the variance of $\hat\tau-\tau^{\rm is}$ is zero, but the variance of $\tau^{\rm is} - \tau^{\rm pop}$ over repeated samples can be large.
We therefore focus on the variance of estimators $\hat\tau$ relative to \(\tau^{\rm is}\) for the particular sample at hand, rather than on the variance of $\hat\tau$ relative to the population average \(\tau^{\rm pop}\). 

If $F_0$ was known, we could estimate $\theta$ efficiently by some version of maximum likelihood to get an estimate $\hat{\theta}$ and 
\begin{equation*}
    \hat{\tau} = \int_0^1 \left( h(F_0^{-1}(y), \hat{\theta}) - F_0 ^{-1}(y) \right) dy
\end{equation*}
as an estimate of $\tau$. The density of $Y$ given $Z$ is 
\begin{equation*}
   f(y|z)= \left(f_0(y)\right)^{1-z} \left(f_1(y, \theta) \right)^z.
\end{equation*}
By Assumption~\ref{as:par-outcomes}
\begin{align}
    f_1(y, \theta) = f_0(h^{-1}(y,\theta)) \frac{\partial h^{-1}(y,\theta)}{\partial y} ,\label{eq:f1(y,theta)}
\end{align}
and the score function is
\begin{equation*}
    \dot{\ell}(y,z,\theta) \equiv z \cdot \frac{\partial}{\partial \theta} \log f_1(y, \theta)
\end{equation*}
yielding,
\begin{equation*}
    \hat{\theta} = \theta + \frac {1}{n} \sum_{i=1}^n I^{-1} \dot{\ell}(Y_i,Z_i,\theta) + o_P(n^{-1/2})
\end{equation*}
where $I=E\left(\dot{\ell}(Y,Z,\theta)\right)^2$.

If $F_0$ is assumed unknown, to obtain an efficient influence function we must 
\begin{itemize}
\item[1)] Compute the tangent plane as $f_0$ varies with $\theta$ fixed. The tangent plane is 
\begin{equation*}
\begin{aligned}
\dot{P}_f = \Bigl\{ & \, u(Y,Z) = (1-Z) v(Y,\theta) + Z v(h^{-1}(Y,\theta),\theta): \\
                    & \qquad  \int v^2(y,\theta)f_0(y)dy <\infty, \int v(y,\theta) f_0(y)dy=0 \Bigr\}.
\end{aligned}
\end{equation*}
(Note that both factors of the likelihood must be varied treating $\theta$ as fixed.)

\item[2)] Project $\dot{\ell}$ on the orthocomplement of the tangent plane to get 
\begin{equation*}
    \dot{\ell}^*(Y,Z,\theta) = \dot{\ell}(Y,Z,\theta) - (Z Q(h^{-1}(Y,\theta), \theta) + (1-Z) Q(Y, \theta))
\end{equation*}
where $Z Q(h^{-1}(Y,\theta), \theta) + (1-Z) Q(Y, \theta)$ is the projection of $\dot{\ell}$ on $ \dot{P}_f$.

\item[3)] The efficient influence function is given by 
\begin{equation*}
    \psi(Y,Z,\theta) = \frac{\dot{\ell}^*(Y,Z,\theta)}{ E\left(\dot{\ell}^*(Y,Z,\theta)\right)^2.}
\end{equation*}
\end{itemize}

\begin{lemma}\label{thm:eif-general} 
The efficient influence function for $\theta$ is
\begin{equation*}
\psi_{f_0}(y, z, \theta )=I^{-1}
\left\{
\frac{z}{p}\cdot g(y,\theta ) - \frac{1-z}{1-p} \cdot g(h(y,\theta ), \theta )
\right\},
\end{equation*}
where
\begin{align*}
g(y, \theta )  &= \frac{\partial}{\partial \theta } \log f_1(y,\theta)= \frac{\partial}{\partial \theta } \log\left(f_0(h^{-1}(y,\theta))\cdot \frac{\partial h^{-1}(y,\theta)}{\partial y}\right),
\end{align*}
and 
\begin{align*}
    I = \int g^2(h(y,\theta ),\theta )f_0(y)dy.
\end{align*}
\end{lemma}

To use the influence function approach we need a $\sqrt n$-consistent initial estimator $\tilde\theta$. We can do so by a fixed number of quantiles, $u_1,\ldots,u_d$, where $d$ is the dimension of $\theta$, and find the $\theta$ that solves
 \[ \hat F^{-1}_1(u)=h(\hat F_0^{-1}(u),\theta),
\]
for $u=u_1,\ldots,u_d$. 
For simplicity, we suggest using evenly-spaced quantiles, $u=\frac{1}{1+d},\frac{2}{1+d},\dots,\frac{d}{1+d}$.
Let $\tilde\theta$ denote the solution to this system of equations.
Then:
\begin{theorem}\label{thm:eif-est-general}
Under mild conditions on the estimation of density and its derivative, the estimator $\hat{\theta}^{if}$ below
is efficient for $\theta$, i.e. $\sqrt{n}(\hat{\theta} - \theta) \stackrel{d}{\Rightarrow} \mathcal{N}(0,\frac{1}{p(1-p)} I^{-1})$:
\begin{align*}
\hat{\theta}^{if} & \equiv\tilde{\theta}+\frac{1}{n}\left\{\sum_{i=1}^{n/2}\psi_{\hat{f}_{0(2)}}(Y_i,Z_{i};\tilde{\theta})+\sum_{i=1+n/2}^{n}\psi_{\hat{f}_{0(1)}}(Y_i,Z_{i}; \tilde{\theta})\right\} 
\end{align*}
where $\hat{f}_{0(1)}$ is the estimate of $f_0$ using $\{{(Y_i,Z_i): 1\leq i \leq \frac{n}{2}}\}$, and $\hat{f}_{0(2)}$ is the estimate of $f_0$  using $\{(Y_i,Z_i): \frac{n}{2} < i \leq  n\}$, again a one step estimate using the sample splitting technique \citep{klaassen1987consistent}, similar to \eqref{eq:tau-if}.
\end{theorem}

Note that, unlike the constant treatment effect setting, the efficient influence function is not necessarily the one corresponding to $F_0$ known.

Given inference for $\hat\theta$, inference for $\hat\tau$ as an estimator of the in-sample average treatment effect,  is straightforward based on the Delta method and the representation in (\ref{tau_is}), taken as given the potential outcomes. The asymptotic variance for $\hat\tau$ is equal to the variance  for $\hat\theta$, pre and post multiplied by the derivative of the expression in  (\ref{tau_is}) with respect to $\theta.$

\section{Simulations}\label{section:simulations}

We evaluate the performance of the proposed estimators and conventional estimators in a Monte Carlo study.
Throughout most of these simulations, the true unit-level treatment effects are all zero.
We estimate the treatment effect using the proposed efficient estimators based on an additive model.
We consider seven estimators:
The (standard) difference in means, the difference in medians, the Hodges-Lehman (\citet{hodges1963estimates}) estimator,\footnote{The Hodges-Lehmann estimator is equal to the median of all pairwise differences between treated and control observations.}
the adaptively trimmed mean, the adaptively winsorized mean,\footnote{We apply the ideas of \citet{jaeckel1971some} for the optimal trimmed mean to choose the parameters for the estimators in Section~\ref{section:partial}, see Theorem A.3 in the Appendix. We allow anywhere between no trimming (difference in means) and the extreme of trimming all but the medians (difference in medians). While including the extremes is not covered by the theory, this approach appeared to work best in our simulations.} the estimator based on the efficient influence function (eif), and the weighted average quantiles (waq) estimator. For the latter two we report only results without sample splitting. Results for the case with sample splitting are very similar and are available in the supplementary materials.
Although in our illustrations we use relatively simple  estimators for the densities and their derivatives based on variable bandwidth kernels, an alternative would be to use methods directy aimed at estimating derivatives  of the logarithm of the density as in \citet{pinkse2021estimates}.
Detailed descriptions of how the new estimators are implemented, as well as R code implementing all estimators with performance optimizations for these simulations, are available in the supplemental materials.\footnote{The fully documented R package  is available at \url{https://github.com/michaelpollmann/parTreat}. Details on the empirical implementation of our estimators are in the supplementary materials. For a sample of 1,000 treated and 1,000 control observations, the R package computes estimates and standard errors practically instantaneously.
With very large samples, the derivatives of the log density can be precomputed on a random subsample of the data for similarly fast computation.}
We present results for three sets of simulations, one with   a range of known distributions for the potential outcomes, so we can directly assess the ability of the proposed methods to adapt to different distributions, and two with simulations based on real data: one  based on housing prices and one based on medical expenditures, both
with thick tailed distribution.

\subsection{Simulations with Known Distributions}

We simulate samples of \(n=\)20,000 observations, half of which are treated, and report summary statistics based on 10,001 simulated samples (using an odd number so that the median is unique).
We repeat the simulation study for standardized Normal, Double Exponential (Laplace), and Cauchy distributions for the potential outcomes.
The difference in means is the maximum likelihood estimator for Normally distributed data, and so will do well there, but may perform poorly for thicker tailed distributions such as the Double Exponential distribution and in particular 
 the Cauchy distribution.
The difference in medians is the maximum likelihood estimator for the Double Exponential distribution, and relatively robust to thick tails and outliers, and so  is  expected to perform reasonably well across all specifications, but not as well as the efficient estimators for the Normal.

For the simulations with known distributions we can derive the functional form for the optimal weights for the quantile-based estimator.
The optimal weights for the waq estimator are proportional to the (estimated) second derivative of the log density.
For the Normal distribution, \(\pder[2]{\ln f}{y}(y) = - \frac{1}{\sigma^2}\), implying the optimal weights are constant.
The density of the double exponential distribution is such that the optimal
 weights asymptotically place all weight close to the median.
For the standard Cauchy distribution the efficient weights \(w\) on the difference in \(u\in(0,1)\) quantiles of treated and control distributions are \(w_f(u) \propto - \cos(2 \pi u) \sin(\pi u)^2\), shown in Figure~\ref{fig:weight-cauchy}.
Most of the weight is concentrated around the median, with strictly {negative} weights  outside the $[0.25,0.75]$  quantile range.

\begin{figure}
\centering
\includegraphics[width=0.49\linewidth]{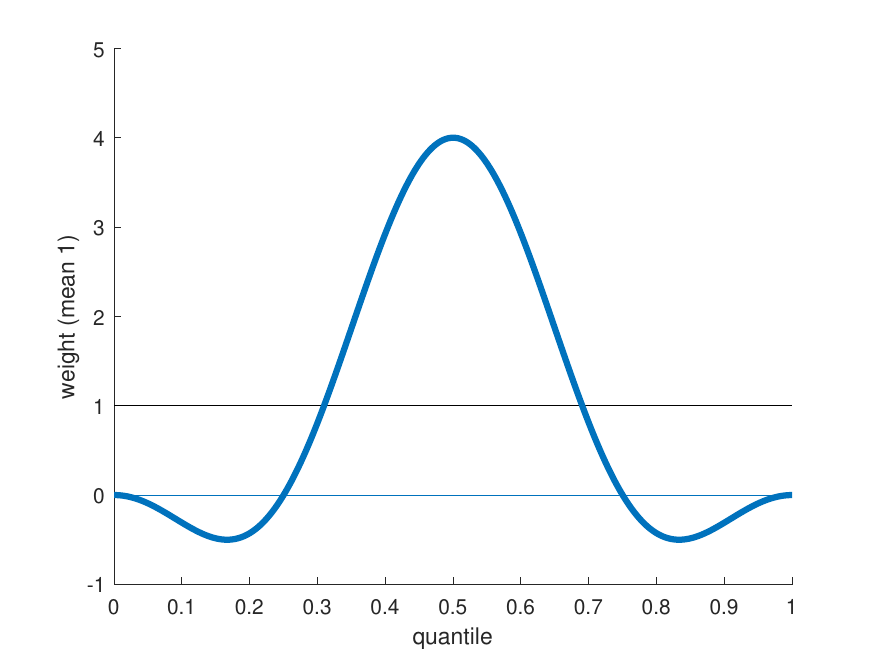}
\hfill
\includegraphics[width=0.49\linewidth]{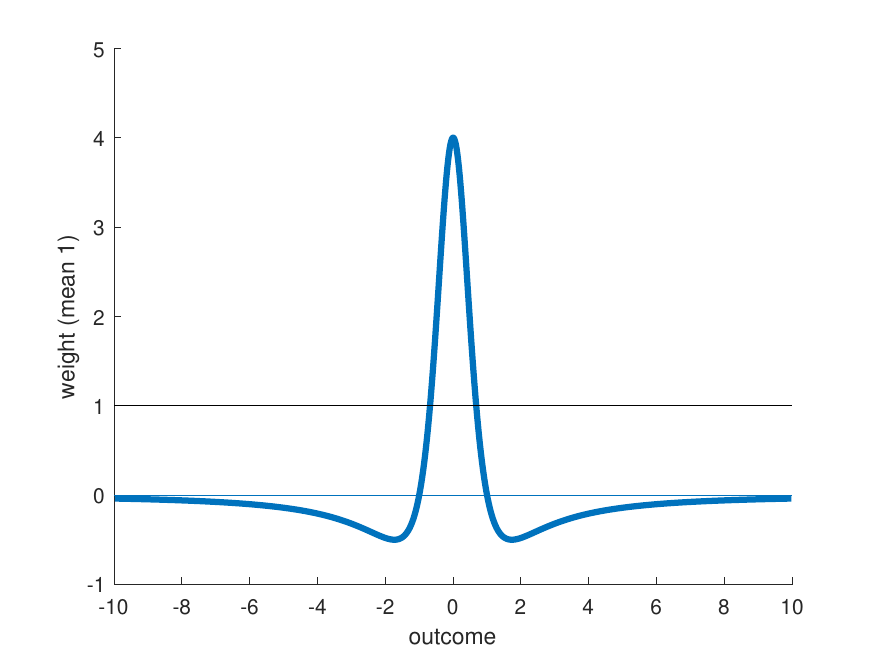}
\caption{\label{fig:weight-cauchy}The efficient weight function for the Cauchy distribution.
The weights, normalized to have mean 1, are plotted against the quantile (left) and against the value of the observations (right).
For the figure on the right, we only plot the range from \(-10\) to \(10\), which corresponds to approximately the \(0.03\) to \(0.97\) quantile.
At this point, the weight is approximately \(-0.04\), and weights for more extreme quantiles are closer to 0.}
\end{figure}

The efficient estimators perform well across distributions, and confidence intervals based either on estimates of the analytic variance formulas or on the bootstrap achieve their nominal coverage levels, as shown in Table~\ref{tab:simulations-parametric}.
Their standard deviations are close to the theoretical efficiency bound, as shown in column 4, labelled relative efficiency, where values larger than one imply standard deviations of the estimator in excess of the efficiency bound.
The efficient influence function and weighted average quantiles estimators are close to the most efficient estimator for the normal, double exponential, and Cauchy distribution.
The last columns show that the confidence intervals are close to their nominal coverage for each distribution.
For computational convenience in the simulations, the confidence intervals based on the bootstrap variance use the \(m\)-out-of-\(n\) bootstrap (\citet{bickel2012resampling}), with \(m= \)2,000 (half treated, half control), to estimate the variance of the estimators.
Even with these smaller sample sizes, the density estimates calculated within each bootstrap sample appear to be sufficiently good to yield reasonable confidence intervals for the estimators.

\begin{table}
    \caption{\label{tab:simulations-parametric}Summary statistics 
    for simulations with different distributions: 
    Normal, Double Exponential, Cauchy. Statistics 
    shown are based on 10,001 simulated samples. 
    Each sample has 10,000 treated observations 
    and 10,000 control observations. Relative 
    efficiency is the ratio of standard deviation 
    to the square root of the efficiency bound. 
    Columns labeled length show the median length 
    of the confidence intervals.}
    \centering
\resizebox{1\linewidth}{!}{
\begin{tabular}{l rrrrr rr rr}
\toprule
& & & & & & \multicolumn{2}{c}{95\% C.I.} & \multicolumn{2}{c}{boot. var. C.I.} \\
\cmidrule(lr){7-8} \cmidrule(lr){9-10}
estimator & bias & \makecell{standard \\ deviation} & \makecell{relative \\ efficiency} & RMSE & MAD & coverage & \makecell{median \\ length} & coverage & \makecell{median \\ length}  \\
\midrule
\multicolumn{10}{l}{Normal Distribution:} \\
diff. in means & 0.000 & 0.014 & 1.01 & 0.014 & 0.010 & 0.95 & 0.055 & 0.95 & 0.055 \\ 
diff. in medians & -0.000 & 0.018 & 1.26 & 0.018 & 0.012 & 0.95 & 0.069 & 0.95 & 0.069 \\ 
Hodges-Lehmann & 0.000 & 0.015 & 1.03 & 0.015 & 0.010 & 0.95 & 0.057 & 0.95 & 0.057 \\ 
adaptive trim & 0.000 & 0.015 & 1.03 & 0.015 & 0.010 & 0.94 & 0.055 & 0.95 & 0.057 \\ 
adaptive wins. & 0.000 & 0.014 & 1.01 & 0.014 & 0.010 & 0.95 & 0.055 & 0.95 & 0.055 \\ 
eif & 0.000 & 0.014 & 1.02 & 0.014 & 0.010 & 0.95 & 0.056 & 0.95 & 0.057 \\ 
waq & 0.000 & 0.014 & 1.02 & 0.014 & 0.010 & 0.95 & 0.056 & 0.95 & 0.056 \\ 
\\
\multicolumn{10}{l}{Double Exponential Distribution:} \\
diff. in means & -0.000 & 0.020 & 1.43 & 0.020 & 0.014 & 0.95 & 0.078 & 0.95 & 0.078 \\ 
diff. in medians & 0.000 & 0.014 & 1.01 & 0.014 & 0.010 & 0.97 & 0.061 & 0.95 & 0.057 \\ 
Hodges-Lehmann & -0.000 & 0.016 & 1.17 & 0.016 & 0.011 & 0.95 & 0.064 & 0.95 & 0.064 \\ 
adaptive trim & 0.000 & 0.014 & 1.02 & 0.014 & 0.010 & 0.95 & 0.059 & 0.95 & 0.057 \\ 
adaptive wins. & -0.000 & 0.018 & 1.29 & 0.018 & 0.012 & 0.96 & 0.076 & 0.95 & 0.069 \\ 
eif & -0.000 & 0.015 & 1.06 & 0.015 & 0.010 & 0.95 & 0.060 & 0.96 & 0.060 \\ 
waq & -0.000 & 0.015 & 1.07 & 0.015 & 0.010 & 0.95 & 0.060 & 0.96 & 0.061 \\ 
\\
\multicolumn{10}{l}{Cauchy Distribution:} \\
diff. in means & 0.462 & 127.149 & 6357.47 & 127.144 & 2.047 & 0.98 & 9.324 & 0.98 & 9.292 \\ 
diff. in medians & 0.000 & 0.022 & 1.11 & 0.022 & 0.015 & 0.97 & 0.093 & 0.95 & 0.087 \\ 
Hodges-Lehmann & 0.000 & 0.026 & 1.28 & 0.026 & 0.018 & 0.95 & 0.101 & 0.95 & 0.101 \\ 
adaptive trim & 0.000 & 0.022 & 1.08 & 0.022 & 0.015 & 0.95 & 0.092 & 0.95 & 0.086 \\ 
adaptive wins. & 0.000 & 0.024 & 1.19 & 0.024 & 0.016 & 0.98 & 0.111 & 0.97 & 0.100 \\ 
eif & 0.000 & 0.020 & 1.01 & 0.020 & 0.014 & 0.97 & 0.085 & 0.97 & 0.088 \\ 
waq & 0.000 & 0.021 & 1.05 & 0.021 & 0.014 & 0.96 & 0.085 & 0.98 & 0.099 \\
\bottomrule
\end{tabular}
}
\end{table}

\subsection{Simulations with House Price Data}
\label{section:applications}

In the second set of simulations 
we use house price data from the replication files of \citet{linden2008estimates} available at \citet{linden2019replication}.
They obtained property sales data for Mecklenburg County, North Carolina, between January 1994 and December 2004.
They dropped sales below \$5,000 and above \$1,000,000, such that 170,239 observations remain, which we take as our population of interest.
Despite the trimming  the distribution is noticeably skewed (skewness \(2.2\)) and thick tailed (kurtosis \(9.5\)).
Even after taking logs, the distribution is heavy-tailed with kurtosis equal to \(5.1\).
Figure~\ref{fig:house-prices-weights} plots a histogram for house prices, both in levels and in logs, along with the estimated optimal weights (minus the second derivative of the log density) based on all 170,239 observations.

\begin{figure}
\centering
\includegraphics[width=0.49\linewidth]{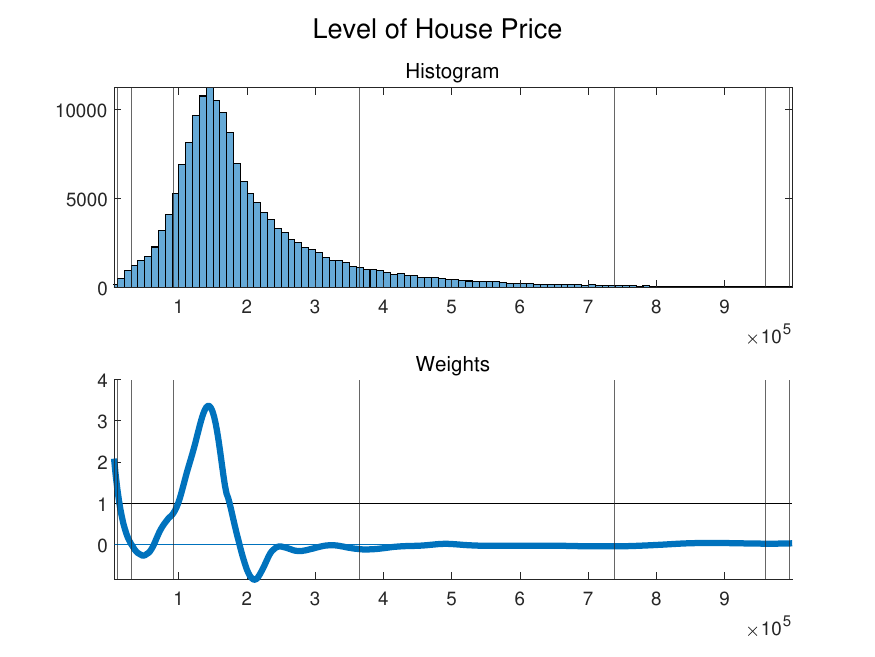}
\hfill
\includegraphics[width=0.49\linewidth]{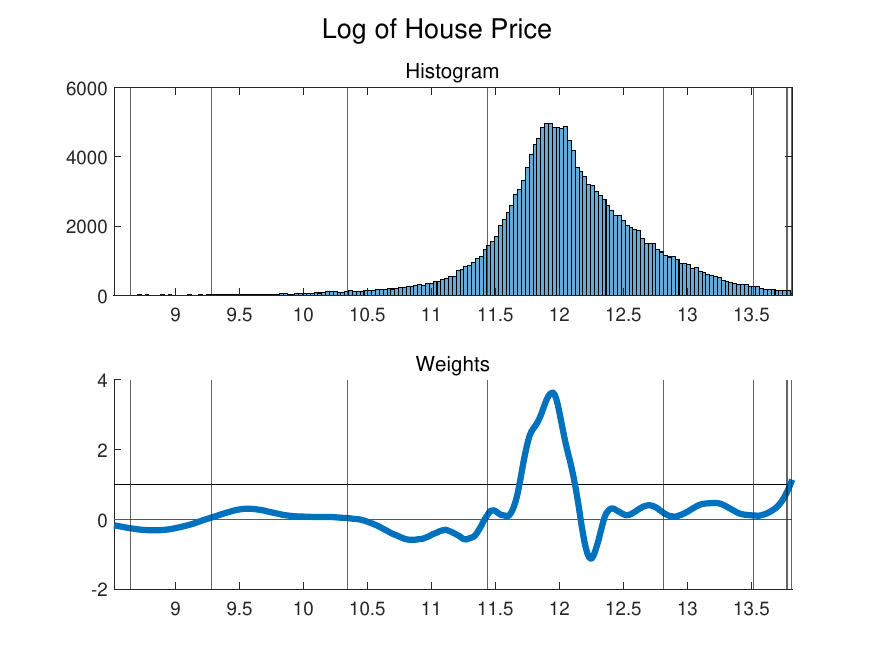}
\caption{\label{fig:house-prices-weights}Histogram of house prices, in levels and logs, as well as estimated optimal weights (minus the second derivative of the log density) based on all 170,239 observations.
The weights are normalized to be mean 1 (black horizontal line). Some weights are below 0 (blue horizontal line).
Vertical lines indicate the \(0.0001\), \(0.001\), \(0.01\), \(0.1\), \(0.9\), \(0.99\), \(0.999\), \(0.9999\) quantiles.}
\end{figure}

We base simulations on this data by drawing samples of size \(n=\)20,000, and randomly assigning exactly half of each sample to the treatment group and the remaining half to the control group, with a zero treatment effect.
Within each sample, observations are drawn from the population without replacement, but sampling is independent across samples, such that observations may appear in multiple samples. We estimate the efficiency bound using density estimates based on all observations.
We compute the same estimators as in the simulations of the previous section, with a small adjustment to the adaptively trimmed and winsorized means where we
 fix the trimming and winsorizing percentiles on the left to 0\% (no trimming/winsorizing), and only adaptively choose the threshold on the right.

Table~\ref{tab:house-prices-simulations} summarizes the simulation results based on 10,001 simulated samples.
When the house prices are in levels, the standard deviation of the difference in means estimator is twice as large as that of efficient estimators.
For the difference in medians and the Hodges-Lehman estimators, which are less affected by outliers in the data, the standard deviation is larger by approximately 30\% and 20\%, respectively.
Confidence intervals, based on estimated variances and asymptotic normal approximations, have close to nominal coverage throughout, and are meaningfully shorter for the efficient estimators we propose.

In Figure~\ref{fig:house-prices-heterogeneous-effects} we show the root mean squared error and coverage of 95\% confidence intervals \emph{both relative to the average treatment effect} under deviations from the constant treatment effect model.\footnote{The design of these simulations was kindly suggested by one of the referees.}
In the top panel, the unit-level treatment effects are independent draws from a normal distribution with mean equal to 0.1 standard deviations of the (population) standard deviation of the potential outcomes in the absence of treatment.
On the horizontal axis, we vary the standard deviation of the normal distribution as a fraction \(q\) of the (population) standard deviation of the control potential outcomes, simulating 10,001 samples for each value.
When \(q>0\), the constant treatment effect model is misspecified.
In the bottom panel, the unit-level treatment effects are \(0\) with probability \(q\) and \(t\) with probability \(1-q\), where \(t\) is chosen as a function of \(q\) such that the average treatment effect is constant across all simulations and the same as in the top panel.
The difference in means estimator is unbiased for the average treatment effect regardless of the value of \(q\), so its large root mean squared error is due to its variance.
In both panels, the influence function-based and weighted average quantile estimators are only (asymptotically) unbiased for the average treatment effect when \(q=0\).

\begin{figure}
\centering
\includegraphics[width=\linewidth]{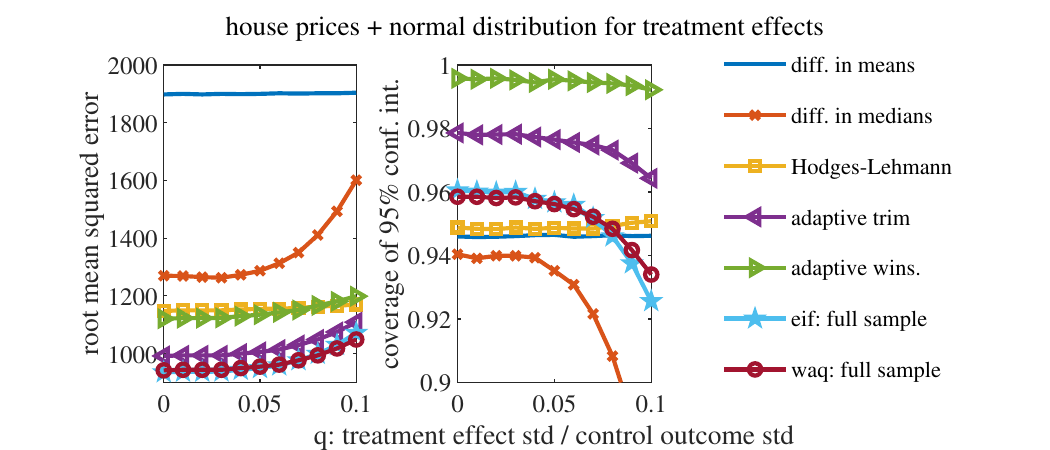}

\vspace{1em}

\includegraphics[width=\linewidth]{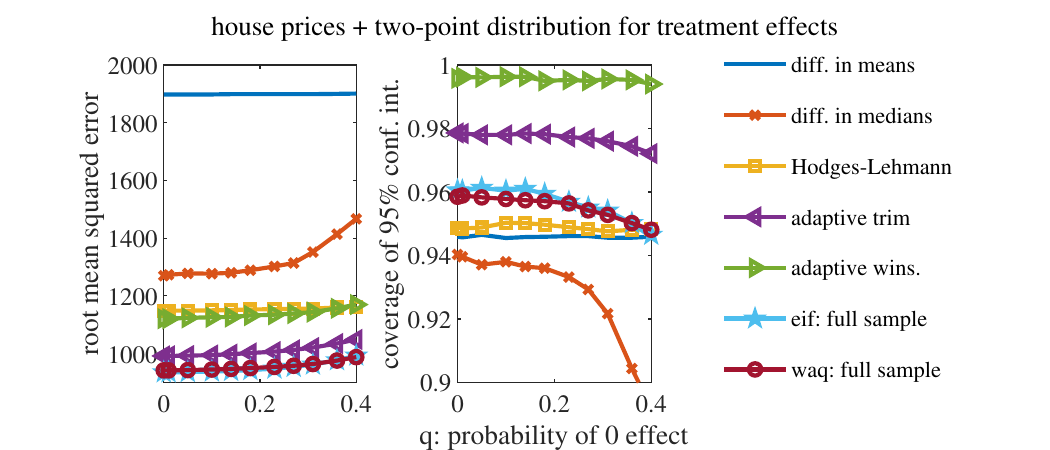}

\caption{\label{fig:house-prices-heterogeneous-effects}Root mean squared error and coverage of 95\% confidence intervals relative to the average treatment effect (fixed at 0.1 standard deviations of the outcome in the absence of treatment) in simulations with heterogeneous treatment effects in the house price data. The horizontal axis varies the amount of heterogeneity. When \(q=0\), there is no heterogeneity such that the constant additive treatment effect model is correctly specified.}
\end{figure}

We also estimate a proportional treatment effect (multiplicative) model and translate the estimated coefficients into level effects.
Under the multiplicative model, \(Y_{i}(1) = \theta Y_{i}(0)\).
When outcomes are strictly positive, this is identical to an additive model for log outcomes; \(\log(Y_{i}(1)) = \log(\theta) + \log(Y_{i}(0))\).
Given an estimate \(\hat{\tau}_{\rm log}\) of the additive model with the outcome in logs, the estimate of the level treatment effect is then
\begin{equation*}
\hat{\tau} = \frac{1}{n} \Bigl(
\sum_{i=1}^{n}
(1-Z_{i}) \bigl((\exp(\hat{\tau}_{\rm log}) - 1) Y_{i} \bigr)
+
Z_{i} \bigl((1 - \exp(-\hat{\tau}_{\rm log})) Y_{i} \bigr)
\Bigr).
\end{equation*}
For estimates of the in-sample treatment effect, we therefore apply estimates \(\hat{\tau}_{\rm log}\) to a population of interest with known means \(\mu_{Y_{0}}\) and \(\mu_{Y_{1}}\) and fixed treatment probability \(p\) as
\begin{equation*}
\hat{\tau} =
(1-p) \bigl((\exp(\hat{\tau}_{\rm log}) - 1) \mu_{Y_{0}} \bigr)
+
p \bigl((1 - \exp(-\hat{\tau}_{\rm log})) \mu_{Y_{1}} \bigr)
.
\end{equation*}
Using the Delta method, if \(V\) is the asymptotic variance of \(\hat{\tau}_{\rm log}\), then the asymptotic variance of \(\hat{\tau}\), holding \(\mu_{Y_{0}}\), \(\mu_{Y_{1}}\), and \(p\) fixed, is
\begin{equation*}
\Bigl((1-p) \exp(\tau_{\rm log}) \mu_{Y_{0}}
+
p \exp(-\tau_{\rm log}) \mu_{Y_{1}}
\Bigr)^{2}
V
\end{equation*}
which we estimate by replacing \(\tau_{\rm log}\) with \(\hat{\tau}_{\rm log}\) and \(V\) by the estimate of the variance, \(\hat{V}\).
For the purpose of these simulations, we set \(\mu_{Y_{0}} = \mu_{Y_{1}}\) equal to the population mean of house prices, and \(p = 1/2\).

When treatment effects are assumed to be proportional to potential outcomes, the proposed estimators for this multiplicative model are still more efficient than alternative estimators, but the gains are smaller.
The middle panel of Table~\ref{tab:house-prices-simulations} shows the quality of estimates of the multiplicative parameter obtained by transforming outcomes into logs, \(\log(\theta)\).
As can be seen in panel~(b) of Figure~\ref{fig:house-prices-weights}, the distribution of log house prices appears closer to the normal distribution with fewer ``outliers.''
Consequently, the difference in means estimator, which is efficient for the normal distribution, comes noticeably closer to the efficiency bound than when outcomes are in levels.
Nevertheless, the efficient estimators further reduce the variance.
The bottom panel of Table~\ref{tab:house-prices-simulations} shows the same summary statistics when the multiplicative parameter is translated back into a level effect.
The efficient estimators of the multiplicative parameter lead to treatment effects with smaller variance and shorter confidence intervals than the difference in means, irrespective of whether the latter is estimated in levels (top panel) or in logs and then translated into levels (bottom panel).

\begin{table}
    \caption{\label{tab:house-prices-simulations}Summary statistics for simulations based on house price data from \citet{linden2008estimates}.
    Statistics shown are based on 10,001 simulated samples.
    Each sample has 10,000 treated observations and 10,000 control observations.
    Relative efficiency is the ratio of standard deviation to the square root of the efficiency bound, which we calculate from density estimates based on the full sample.
    Columns labeled length show the median length of the confidence intervals.
    The adaptively trimmed mean and adaptively winsorized mean only apply trimming at the top, but not at the bottom.
    }
    \centering
\resizebox{1\linewidth}{!}{
    \begin{tabular}{l rrrrr rr rr}
    \toprule
    & & & & & & \multicolumn{2}{c}{95\% C.I.} & \multicolumn{2}{c}{boot. var. C.I.} \\
    \cmidrule(lr){7-8} \cmidrule(lr){9-10}
    estimator & bias & \makecell{standard \\ deviation} & \makecell{relative \\ efficiency} & RMSE & MAD & coverage & length & coverage & length \\
    \midrule
    \multicolumn{10}{l}{effect in levels based on additive model in levels} \\
    diff. in means & -28 & 1871 & 1.92 & 1871 & 1251 & 0.95 & 7343 & 0.95 & 7339 \\ 
    diff. in medians & 8 & 1259 & 1.30 & 1259 & 855 & 0.95 & 4989 & 0.95 & 4893 \\ 
    Hodges-Lehmann & -5 & 1138 & 1.17 & 1138 & 757 & 0.95 & 4487 & 0.95 & 4487 \\ 
    adaptive trim & 5 & 989 & 1.02 & 989 & 651 & 0.98 & 4560 & 0.95 & 3914 \\ 
    adaptive wins. & 4 & 1114 & 1.15 & 1114 & 738 & 0.99 & 6378 & 0.97 & 4760 \\ 
    eif & 13 & 941 & 0.97 & 941 & 628 & 0.95 & 3819 & 0.95 & 3768 \\ 
    waq & 13 & 946 & 0.97 & 946 & 626 & 0.95 & 3819 & 0.95 & 3859 \\ 
    \\
    \multicolumn{10}{l}{multiplicative parameter: additive model in logs} \\
    diff. in means & -0.0000 & 0.0083 & 1.35 & 0.0083 & 0.0055 & 0.95 & 0.033 & 0.95 & 0.032 \\ 
    diff. in medians & 0.0000 & 0.0075 & 1.23 & 0.0075 & 0.0051 & 0.95 & 0.030 & 0.95 & 0.029 \\ 
    Hodges-Lehmann & -0.0000 & 0.0072 & 1.17 & 0.0072 & 0.0048 & 0.95 & 0.028 & 0.95 & 0.028 \\ 
    adaptive trim & -0.0000 & 0.0083 & 1.35 & 0.0083 & 0.0055 & 0.95 & 0.033 & 0.95 & 0.032 \\ 
    adaptive wins. & -0.0000 & 0.0083 & 1.35 & 0.0083 & 0.0055 & 0.95 & 0.033 & 0.95 & 0.032 \\ 
    eif & -0.0000 & 0.0067 & 1.08 & 0.0067 & 0.0045 & 0.95 & 0.026 & 0.95 & 0.026 \\ 
    waq & -0.0000 & 0.0067 & 1.08 & 0.0067 & 0.0044 & 0.95 & 0.026 & 0.95 & 0.027 \\ 
    \\
    \multicolumn{10}{l}{effect in levels based on additive model in logs} \\
    diff. in means & -7 & 1695 & 1.35 & 1695 & 1125 & 0.95 & 6658 & 0.95 & 6657 \\ 
    diff. in medians & 9 & 1545 & 1.23 & 1545 & 1048 & 0.95 & 6128 & 0.95 & 6001 \\ 
    Hodges-Lehmann & -6 & 1468 & 1.17 & 1468 & 976 & 0.95 & 5786 & 0.95 & 5783 \\ 
    adaptive trim & -7 & 1695 & 1.35 & 1695 & 1125 & 0.95 & 6658 & 0.95 & 6657 \\ 
    adaptive wins. & -7 & 1695 & 1.35 & 1695 & 1125 & 0.95 & 6658 & 0.95 & 6657 \\ 
    eif & -5 & 1363 & 1.08 & 1363 & 914 & 0.95 & 5374 & 0.95 & 5415 \\ 
    waq &-3 & 1364 & 1.08 & 1364 & 910 & 0.95 & 5374 & 0.95 & 5462 \\
    \bottomrule
    \end{tabular}
}
\end{table}

\subsection{Medical Expenditures Data}

Next, we present simulation results based on confidential medical expenditure data from the IBM MarketScan Research Database, following the sample construction of \citet[Figure~2]{koenecke2021alpha}.
We restrict the sample to males, age 45--64, with pneumonia inpatient diagnosis and at least 1 year of continuous medical enrollment.
For each patient, we consider the first inpatient admission only to abstract away from any dynamics.
We focus on medical expenditure 
as the outcome variable.
For each patient, we sum the payments recorded by MarketScan for this admission.
In total, we use data on 103,662 admissions.\footnote{The sample size differs slightly from that reported by \citet{koenecke2021alpha} due to missing expenditure data for a small number of admissions.}
Figure~\ref{fig:medical-expenditures-weights} plots a histogram for medical expenditure in levels and in logs along with the estimated optimal weights (minus the second derivative of the log density).
We also observe a treatment variable in this data set,
the (prior) use of alpha blockers, which \citet{koenecke2021alpha} find may improve health outcomes during respiratory distress by preventing hyperinflammation.

\begin{figure}
\includegraphics[trim={0 0 0 2.3in},clip,width=0.49\linewidth]{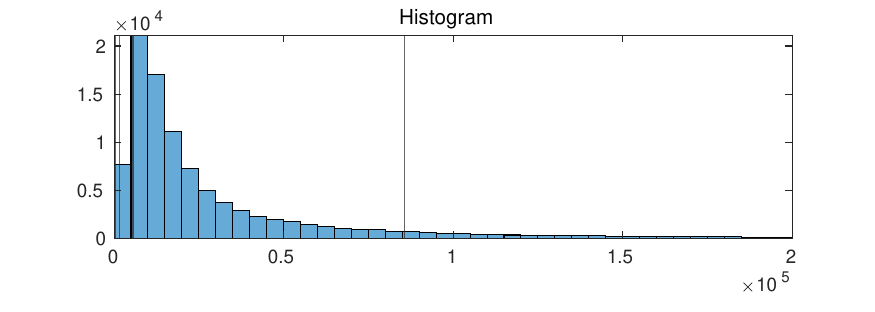}
\includegraphics[trim={0 2.05in 0 0.4in},clip,width=0.49\linewidth]{fig_weight_level_health-2021-06-01.pdf}
\hfill
\includegraphics[trim={0 0 0 2.3in},clip,width=0.49\linewidth]{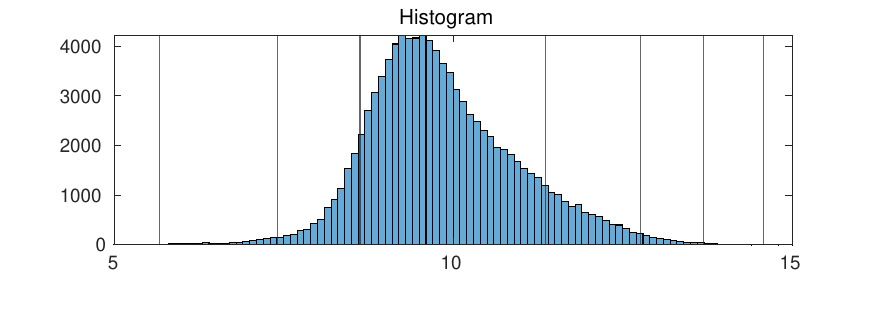}
\includegraphics[trim={0 2.05in 0 0.4in},clip,width=0.49\linewidth]{fig_weight_log_health-2021-06-01.pdf}
\caption{\label{fig:medical-expenditures-weights}Histogram of medical expenditures per admission, in levels and logs, as well as estimated optimal weights (minus the second derivative of the log density), based on the 98,155 observations in the control group.
For the figure in levels, vertical lines indicate the \(0.001\), \(0.01\), \(0.1\), and \(0.9\) quantiles; the figure is limited to below \$200,000, such that the \(0.99\) and higher quantiles do not appear.
For the figure in logs, vertical lines indicate the \(0.001\), \(0.01\), \(0.1\), \(0.99\), \(0.999\), \(0.9999\) quantiles.
The weights are normalized to be mean 1 (black horizontal line).
Some weights are below 0 (blue horizontal line).}
\end{figure}

We design a simulation study similar to those of the previous sections, treating the receipt of alpha blockers as randomly assigned in our population.
This allows us to study (coverage) properties of the estimators and inference procedures in settings where the parametric treatment effect model is not (necessarily) correctly specified.
For these simulations, each sample is a draw, without replacement, of 200 of the 5,507 observations in the treatment group and 3,565 of the 98,155 observations in the control group.
While the control group is smaller than in our other simulations, it remains sufficiently large to estimate the density and its derivatives required for our estimators.
In this simulation design, \(F_{0}\) is given by the empirical distribution of the control group, and \(F_{1}\) is given by the empirical distribution of the treatment group, of the full sample.
Although it is not necessarily correct, the additive treatment effect model may offer a reasonable approximation, and we are interested in the performance of inference methods when the conditions for our theoretical results are not (quite) met in this particular application.
The simulation results reported in Table~\ref{tab:health-simulations-log}, for the same estimators as in Section~\ref{section:applications}, suggest that the proposed estimators perform reasonably well in this setting despite mis-specification.

\begin{table}
\caption{\label{tab:health-simulations-log} Summary statistics for simulations based on the medical expenditures data in logs, across 1,000 simulated samples.
Each sample has 200 treated observations and 3,565 control observations.
Coverage and median length refer to 95\% confidence intervals.}
\centering
    \begin{tabular}{l rrrrr rr rr}
    \toprule
    & & & & & \multicolumn{2}{c}{95\% C.I.} & \multicolumn{2}{c}{boot. var. C.I.} \\
    \cmidrule(lr){6-7} \cmidrule(lr){8-9}
    estimator & bias & \makecell{standard \\ deviation} & RMSE & MAD & coverage & length & coverage & length \\
    \midrule
    
\multicolumn{9}{l}{effect in levels based on additive model in levels} \\
    \multicolumn{9}{l}{error and coverage relative to population difference in means} \\
    diff. in means & -108 & 4567 & 4566 & 3016 & 0.93 & 16935 & 0.93 & 16888 \\
    diff. in medians & 3345 & 1271 & 3578 & 3222 & 0.46 & 6161 & 0.28 & 5259 \\
    Hodges-Lehmann & 3490 & 891 & 3602 & 3464 & 0.03 & 3664 & 0.01 & 3534 \\
    adaptive trim & 3877 & 641 & 3929 & 3860 & 0.00 & 3052 & 0.00 & 2778 \\
    adaptive wins. & 2816 & 1384 & 3138 & 2786 & 0.52 & 5668 & 0.57 & 6023 \\
    eif & 3928 & 682 & 3987 & 3922 & 0.01 & 4878 & 0.00 & 3240 \\
    waq & 3846 & 792 & 3927 & 3848 & 0.03 & 4878 & 0.01 & 4018 \\
    \\
\multicolumn{9}{l}{multiplicative parameter: additive model in logs} \\
    \multicolumn{9}{l}{error and coverage relative to population difference in means of outcomes in logs} \\
    diff. in means & -0.001 & 0.077 & 0.077 & 0.052 & 0.95 & 0.30 & 0.95 & 0.30 \\
    diff. in medians & 0.005 & 0.080 & 0.080 & 0.051 & 0.97 & 0.33 & 0.95 & 0.32 \\
    Hodges-Lehmann & 0.008 & 0.071 & 0.071 & 0.048 & 0.95 & 0.29 & 0.95 & 0.28 \\
    adaptive trim & 0.016 & 0.074 & 0.076 & 0.053 & 0.95 & 0.29 & 0.95 & 0.29 \\
    adaptive wins. & -0.002 & 0.078 & 0.078 & 0.052 & 0.94 & 0.30 & 0.95 & 0.31 \\
    eif & 0.029 & 0.066 & 0.072 & 0.049 & 0.95 & 0.27 & 0.94 & 0.26 \\
    waq & 0.028 & 0.066 & 0.071 & 0.047 & 0.95 & 0.27 & 0.95 & 0.27 \\
    \\
\multicolumn{9}{l}{effect in levels based on additive model in logs} \\
    \multicolumn{9}{l}{error and coverage relative to population difference in means} \\
    diff. in means & 2423 & 2871 & 3756 & 2522 & 0.90 & 11167 & 0.91 & 11138 \\ 
    diff. in medians & 2647 & 3009 & 4006 & 2499 & 0.92 & 12318 & 0.92 & 11951 \\ 
    Hodges-Lehmann & 2747 & 2673 & 3832 & 2672 & 0.88 & 10730 & 0.87 & 10390 \\ 
    adaptive trim & 3062 & 2809 & 4154 & 2983 & 0.86 & 11055 & 0.85 & 10850 \\ 
    adaptive wins. & 2365 & 2915 & 3752 & 2544 & 0.90 & 11084 & 0.92 & 11389 \\ 
    eif & 3532 & 2525 & 4341 & 3481 & 0.78 & 10452 & 0.76 & 10014 \\ 
    waq & 3464 & 2523 & 4285 & 3413 & 0.79 & 10441 & 0.77 & 10147 \\
    \bottomrule
    \end{tabular}
\end{table}

\section{Conclusion}
\label{section:conclusion}

In many modern settings where randomized experiments are used to estimate treatment effects the presence of heavy-tailed distributions can lead to larger standard errors. 
Often researchers use winsorizing with {\it ad hoc} thresholds to address this. 
Here we develop systematic methods for obtaining more precise inferences using parametric models for the treatment effects, while avoiding the specification of models for the potential outcomes. We present results for semiparametric effiency bounds, suggest efficient estimators, and show in simulations that these methods can be effective in realistic settings. 

In particular we recommend the semiparametrically efficient estimator under the constant additive treatment effect model. Although one may not think the constant additive treatment effect assumption holds exactly, the fact that the estimator can be interpreted as estimating a weighted average of the quantile treatment effects make this an attractive choice.

In this discussion we do not incorporate covariates or pretreatment variables. One could combine the ideas exposited in the current paper with models for the control potential outcome. One may also wish to incorporate covariates in the model for the quantile treatment effects and thus allow for heterogenous treatment effects  {\it e.g., } \citet{chen2022robust, wager2018estimation}. Another interesting avenue is to consider alternative estimators for the current model by first estimating the unrestricted quantile functions, followed by minimum distance methods, as in \citet{alvarez2022inference, alvarezsemiparametric}.

\newpage

\emph{Conflict of interest:} We have no conflicts of interest to disclose.

\section*{Data availability}

The source for the house price data is \citet{linden2019replication}. The relevant columns of the data are also included in the supplementary materials of this paper.
The medical data are proprietary and confidential. Researchers at some institutions, in particular medical schools, may have access to the IBM MarketScan Research Database from which the data is drawn following \citet[Figure 2]{koenecke2021alpha}.

\section*{Funding}

Generous support from the Office of Naval Research through ONR grants N00014-17-1-2131 and N00014-19-1-2468 is gratefully acknowledged. Data for the health application were accessed using the Stanford Center for Population Health Sciences Data Core.

\newpage

\bibliographystyle{chicago}

\newpage

\section*{List of Figure Legends}

\paragraph{Figure 1.} Example members of the generalized Huber family of distributions.

\paragraph{Figure 2.} The efficient weight function for the Cauchy distribution. The weights, normalized to have mean 1, are plotted against the quantile (left) and against the value of the observations (right). For the figure on the right, we only plot the range from \(-10\) to \(10\), which corresponds to approximately the \(0.03\) to \(0.97\) quantile. At this point, the weight is approximately \(-0.04\), and weights for more extreme quantiles are closer to 0.

\paragraph{Figure 3.} Histogram of house prices, in levels and logs, as well as estimated optimal weights (minus the second derivative of the log density) based on all 170,239 observations. The weights are normalized to be mean 1 (black horizontal line). Some weights are below 0 (blue horizontal line). Vertical lines indicate the \(0.0001\), \(0.001\), \(0.01\), \(0.1\), \(0.9\), \(0.99\), \(0.999\), \(0.9999\) quantiles.

\paragraph{Figure 4.} Root mean squared error and coverage of 95\% confidence intervals relative to the average treatment effect (fixed at 0.1 standard deviations of the outcome in the absence of treatment) in simulations with heterogeneous treatment effects in the house price data. The horizontal axis varies the amount of heterogeneity. When \(q=0\), there is no heterogeneity such that the constant additive treatment effect model is correctly specified.

\paragraph{Figure 5.} Histogram of medical expenditures per admission, in levels and logs, as well as estimated optimal weights (minus the second derivative of the log density), based on the 98,155 observations in the control group. For the figure in levels, vertical lines indicate the \(0.001\), \(0.01\), \(0.1\), and \(0.9\) quantiles; the figure is limited to below \$200,000, such that the \(0.99\) and higher quantiles do not appear. For the figure in logs, vertical lines indicate the \(0.001\), \(0.01\), \(0.1\), \(0.99\), \(0.999\), \(0.9999\) quantiles. The weights are normalized to be mean 1 (black horizontal line). Some weights are below 0 (blue horizontal line).

\newpage

\appendix

\setcounter{page}{1}\setcounter{section}{0}

\centerline{\sc \Large Supplementary Materials for:}

\vskip0.5cm

\centerline{\sc \Large Semiparametric Estimation of Treatment Effects }
\vskip0.5cm

\centerline{\sc \Large in Randomized Experiments}

\vskip1cm

\centerline{\sc \Large Athey, Bickel, Chen, Imbens \& Pollmann}

\vskip0.5cm 

\centerline{\sc Appendix}

\renewcommand{\thesection}{\Alph{section}}

We first provide a general theorem on linear combination of order statistics, which simplifies the results of \citet{chernoff1967asymptotic}, see also \citet{bickel1967some}, \citet{govindarajulu1967generalizations} and \citet{stigler1974linear}, and then provide the technical details for proving the main results.

\section{A general theorem on linear combination of order statistics}

Let $F$ be a distribution function, twice differentiable, and $f=F'$.
Let $W:[0,1]\rightarrow\mathcal{R}$ be a function s.t. $W(0)=0$
and $W(1)=1$. The C-R condition was first proposed by \citet{Csorgo-Revesz-1978}.
\begin{description}
\item [{(C-R)}] There exists $C_{0}>0$ and $\epsilon\in(0,\frac{1}{2})$
such that
\end{description}
\begin{enumerate}
\item $F(x)(1-F(x))\frac{|f'(x)|}{f^{2}(x)}\leq C_{0}$ if $F(x)\in(0,\epsilon)$
or $F(x)\in(1-\epsilon,1)$, and $C_{0}$ can be replaced by another
constant if $F(x)\in[\epsilon,1-\epsilon]$.
\item $f'(x)\geq0$ for $x<F^{-1}(\epsilon)$ and $f'(x)\leq0$ for $x>F^{-1}(1-\epsilon)$.
\end{enumerate}
\begin{description}
\item [{(W)}] There exists $\epsilon_{n}=o(1)$ such that
\begin{align}
\int_{0}^{\epsilon_{n}}(1+t^{1-C_{0}})|dW(t)| & =o(n^{-1/2})\text{ and }\int_{1-\epsilon_{n}}^{1}(1+(1-t)^{1-C_{0}})|dW(t)|=o(n^{-1/2})\label{eq:F-W-tail}\\
\int_{\epsilon_{n}}^{1-\epsilon_{n}}\frac{1}{f(F^{-1}(t))}dW(t) & =O(\frac{n^{1/4}}{\log n}).\label{eq:F-W-body}
\end{align}
\end{description}
\begin{theorem}\label{thm:quantile-to-mean}
Let $X_{1},\cdots,X_{n}$ be i.i.d. from $F$, and let $\hat{F}$
and $\hat{F}^{-1}$ be the empirical distribution function and empirical
quantile function respectively. Assume that the C-R condition and the above W
condition hold and let
\begin{align}
\psi(x) & =-\int_{0}^{1}\frac{1(F(x)\leq t)-t}{f(F^{-1}(t))}dW(t).\label{eq:W-to-psi}
\end{align}
Then
\begin{align*}
\sqrt{n}\int_{0}^{1}(\hat{F}^{-1}(t)-F^{-1}(t))dW(t) & =\frac{1}{\sqrt{n}}\sum_{i=1}^{n}\psi(X_{i})+o_{P}(1)
\end{align*}
which converges in distribution to $\mathcal{N}(0, \sigma_f^2)$ if further $E\psi^2(X_1)<\infty$, where
\begin{align}
    \sigma_f^2 = \int_0^1\int_0^1 \frac{\min(s,t) - st}{f(F^{-1}(s))f(F^{-1}(t))}dW(s)dW(t).\label{eq:sigma_f^2}
\end{align}
When $\psi'$ exists, we may also represent $W$ as below if it is well defined:
\begin{align}
W(t) & =\frac{\int_{0}^{t}\psi'(F^{-1}(u))du}{\int_{0}^{1}\psi'(F^{-1}(u))du}.\label{eq:psi-to-W}
\end{align}
\end{theorem}

\begin{remark}
Justification of the conditions for a few common scenarios.
\end{remark}
\begin{enumerate}
\item Gaussian: $f(x)=\frac{1}{\sqrt{2\pi}}e^{-x^{2}/2}$, then $\epsilon=0.5$
and $C_{0}=1$ meet the C-R condition.
\item Cauchy: $f(x)=\frac{1}{\pi(1+x^{2})}$, then $\epsilon=0.5$ and $C_{0}=2$ meet the C-R condition.
\item Equal-weight: $W'(t)=1$, then \eqref{eq:F-W-tail} implies $\epsilon_{n}=o(n^{-1/2})$
for $C_{0}\leq1$ and $\epsilon_{n}=o(n^{-\frac{1}{2(2-C_{0})}})$
for $C_{0}\in(1,2)$. So the equal weight applies to Gaussian. This
does not apply to Cauchy, but note that its sample mean does not converge.
\item Median: $W$ puts all mass at $t=\frac{1}{2}$. The W condition holds
for any $f$.
\end{enumerate}

\section{A lemma on the C-R condition}

\begin{lemma}
\label{lem:C-R-bounds}If the C-R condition holds, then for any $t\in(0,\epsilon)$
\begin{align}
\frac{1}{f(F^{-1}(t))} & \leq K_{0}t^{-C_{0}}\label{eq:CR-f-bound}\\
|F^{-1}(t)| & \leq K_{1}t^{1-C_{0}}+K_{2}\label{eq:CR-F^-1-bound}
\end{align}
where $K_{0},K_{1}$ and $K_{2}$ are positive constants, only dependent
on $C_{0}$ and $\epsilon$. Similar results hold for $t\in(1-\epsilon,1)$,
omitted for conciseness.
\end{lemma}

\begin{proof}
By the C-R condition, $f'(x)\geq0$ for any $x<F^{-1}(\epsilon)$, and for any $t\in(0,\epsilon)$,
\begin{align*}
\int_{t}^{\epsilon}\frac{f'}{f^{2}}(F^{-1}(u))du & \leq\int_{t}^{\epsilon}\frac{C_{0}}{u(1-u)}du=C_{0}\log\frac{\epsilon}{1-\epsilon}-C_{0}\log\frac{t}{1-t}\leq C-C_{0}\log t
\end{align*}
where $C=C_{0}\log\frac{\epsilon}{1-\epsilon}$. Since
\begin{align*}
\frac{f'}{f^{2}}(F^{-1}(u)) & =\frac{\partial}{\partial u}\log f(F^{-1}(u)),
\end{align*}
then
\begin{align*}
\log\frac{f(F^{-1}(\epsilon))}{f(F^{-1}(t))} & \leq C-C_{0}\log t,
\end{align*}
which implies \eqref{eq:CR-f-bound} with $K_{0}=\frac{e^{C}}{f(F^{-1}(\epsilon))}$.
Next, \eqref{eq:CR-F^-1-bound} follows from the inequality below:
\begin{align*}
F^{-1}(\epsilon)-F^{-1}(t)=\int_{t}^{\epsilon}\frac{1}{f(F^{-1}(u))}du & \leq\int_{t}^{\epsilon}K_{0}u^{-C_{0}}du=\frac{K_{0}}{1-C_{0}}(\epsilon^{1-C_{0}}-t^{1-C_{0}}).
\end{align*}
\end{proof}

\section{Proof of Theorem \ref{thm:quantile-to-mean}}
\begin{proof}
First of all,
\begin{align*}
\int_{\epsilon_{n}}^{1-\epsilon_{n}}\big|(\hat{F}^{-1}(t)-F^{-1}(t))+\frac{\hat{F}(F^{-1}(t))-t}{f(F^{-1}(t))}\big|\cdot|dW(t)| & \leq\Delta_{n}\cdot\int_{\epsilon_{n}}^{1-\epsilon_{n}}\frac{|dW(t)|}{f(F^{-1}(t))}
\end{align*}
where
\begin{align*}
\Delta_{n} & =\sup_{0<t<1}\big|f(F^{-1}(t))(\hat{F}^{-1}(t)-F^{-1}(t))+(\hat{F}(F^{-1}(t))-t)\big|.
\end{align*}
Under the C-R condition, by Theorem 4 of \cite{Csorgo-Revesz-1978},
we have
\begin{align*}
\Delta_{n} & =o_{P}(n^{-3/4}\log n).
\end{align*}
Under the condition \eqref{eq:F-W-body}, 
\begin{align*}
\int_{\epsilon_{n}}^{1-\epsilon_{n}}\frac{|dW(t)|}{f(F^{-1}(t))} & =O_{P}(\frac{n^{1/4}}{\log n}).
\end{align*}
Thus
\begin{align*}
\int_{\epsilon_{n}}^{1-\epsilon_{n}}(\hat{F}^{-1}(t)-F^{-1}(t))dW(t) & =-\int_{\epsilon_{n}}^{1-\epsilon_{n}}\frac{\hat{F}(F^{-1}(t))-t}{f(F^{-1}(t))}dW(t)+o_{P}(n^{-1/2}).
\end{align*}

Next, since
\begin{align*}
E\big|\int_{0}^{\epsilon_{n}}\frac{\hat{F}(F^{-1}(t))-t}{f(F^{-1}(t))}dW(t)\big| & \leq2\int_{0}^{\epsilon_{n}}\frac{t}{f(F^{-1}(t))}|dW(t)|\\
 & \leq2K_{0}\int_{0}^{\epsilon_{n}}t^{1-C_{0}}|dW(t)|\\
 & =o(n^{-1/2})
\end{align*}
where the last inequality is due to Lemma \ref{lem:C-R-bounds} and
the last equality is due to \eqref{eq:F-W-tail}, then
\begin{align*}
\int_{0}^{\epsilon_{n}}\frac{\hat{F}(F^{-1}(t))-t}{f(F^{-1}(t))}dW(t) & =o_{P}(n^{-1/2}).
\end{align*}
Furthermore, by Lemma \ref{lem:C-R-bounds},
\begin{align*}
\left|\int_{0}^{\epsilon_{n}}F^{-1}(t)dW(t)\right| & \leq K_{1}\int_{0}^{\epsilon_{n}}t^{1-C_{0}}|dW(t)|+K_{2}\int_{0}^{\epsilon_{n}}|dW(t)|=o(n^{-1/2})
\end{align*}
where the last equality is due to \eqref{eq:F-W-tail}. Similarly,
we have
\begin{align*}
\int_{1-\epsilon_{n}}^{1}\frac{\hat{F}(F^{-1}(t))-t}{f(F^{-1}(t))}dW(t) & =o_{P}(n^{-1/2})
\end{align*}
and
\begin{align*}
\int_{1-\epsilon_{n}}^{1}F^{-1}(t)dW(t) & =o(n^{-1/2}).
\end{align*}

Hence, to prove that
\begin{align*}
\int_{0}^{1}(\hat{F}^{-1}(t)-F^{-1}(t))dW(t) & =-\int_{0}^{1}\frac{\hat{F}(F^{-1}(t))-t}{f(F^{-1}(t))}dW(t)+o_{P}(n^{-1/2}),
\end{align*}
it remains to show that 
\begin{align}
\left(\int_{0}^{\epsilon_{n}}+\int_{1-\epsilon_{n}}^{1}\right)\hat{F}^{-1}(t)dW(t) & =o_{P}(n^{-1/2}).\label{eq:hat-F-1-tail}
\end{align}

Let $X_{(1)}\leq\cdots\leq X_{(n)}$ be the order statistics and $U_{(i)}\equiv F(X_{(i)})$,
then $\hat{F}^{-1}(t)=F^{-1}(U_{(\lceil nt\rceil)})$.

If $C_{0}>1$, then by Lemma \ref{lem:C-R-bounds},
\begin{align*}
\left|\int_{0}^{\epsilon_{n}}\hat{F}^{-1}(t)dW(t)\right|=\left|\int_{0}^{\epsilon_{n}}F^{-1}(U_{(\lceil nt\rceil)})dW(t) \right| & \leq(K_{1}+K_{2})\int_{0}^{\epsilon_{n}}(U_{(\lceil nt\rceil)})^{1-C_{0}}|dW(t)|.
\end{align*}
Since $U_{(i)}$ follows the $\text{Beta}(i,n+1-i)$ distribution, one can
verify that
\begin{align*}
E(U_{(i)}^{a}) & =\frac{\Gamma(n+1)\Gamma(i+a)}{\Gamma(n+1+a)\Gamma(i)}.
\end{align*}
Using the bounds for gamma functions \citep{Batir2017}, i.e. for
$x\geq1,$
\begin{align*}
\sqrt{2\pi}x^{x}e^{-x}(x^{2}+\frac{x}{3}+0.0496)^{1/4} & <\Gamma(x+1)<\sqrt{2\pi}x^{x}e^{-x}(x^{2}+\frac{x}{3}+0.056)^{1/4},
\end{align*}
one can get for $i\geq2-a$,
\begin{align*}
C_{1}\cdot(i/n)^{a}\leq\mathbb{E}U_{(i)}^{a} & \leq C_{2}\cdot(i/n)^{a}
\end{align*}
where $C_{1},C_{2}>0$ are some constants. For $1\leq i<2-a$, $EU_{(i)}^{a}=\frac{\Gamma(i+a)}{\Gamma(i)}n^{-a}(1+o(1))$.
These imply that $EU_{(\lceil nt\rceil)})^{1-C_{0}}\leq C_{3}\cdot t^{1-C_{0}}$
for $C_{3}$ large enough, and hence
\begin{align*}
E\left(\left|\int_{0}^{\epsilon_{n}}\hat{F}^{-1}(t)dW(t)\right|\right) & \leq(K_{1}+K_{2})\int_{0}^{\epsilon_{n}}EU_{(\lceil nt\rceil)})^{1-C_{0}}\big|dW(t)\big|\\
 & \leq C_{3}(K_{1}+K_{2})\int_{0}^{\epsilon_{n}}t^{1-C_{0}}|dW(t)|.
\end{align*}
If $0<C_{0}\leq1$, then 
\begin{align*}
\mathbb{E}\left(\left|\int_{0}^{\epsilon_{n}}\hat{F}^{-1}(t)dW(t)\right|\right) & \leq(K_{1}+K_{2})\int_{0}^{\epsilon_{n}}|dW(t)|.
\end{align*}

Therefore, due to \eqref{eq:F-W-tail}, $\mathbb{E}\big|\int_{0}^{\epsilon_{n}}\hat{F}^{-1}(t)dW(t)\big|=o(n^{-1/2})$
for both $C_{0}>1$ and $C_{0}\in(0,1]$, which implies $\int_{0}^{\epsilon_{n}}\hat{F}^{-1}(t)dW(t)=o_{P}(n^{-1/2})$.
Similarly, it can be shown that $\int_{1-\epsilon_{n}}^{1}\hat{F}^{-1}(t)dW(t)=o_P(n^{-1/2})$.

Thus \eqref{eq:hat-F-1-tail} holds, which concludes the proof.
\end{proof}

\section{Proof of Theorem \ref{thm:efficiency}}

Our proof is based on the sample splitting argument, see \citep{klaassen1987consistent}.

\subsection{General conditions for the density estimation}

Let $X_{1},\cdots,X_{n}$ be i.i.d. $F$ with $f\equiv F'$ and $\hat{F}$
be the empirical distribution. Let $\hat{f}$ be an estimate of $f$
from data independent of $\{X_{1},\cdots,X_{n}\}$ such that $I(\hat{f})=I(f)+o_P(1)$ and the following
conditions hold for any $\delta_{n}=O_{P}(n^{-1/2})$:
\begin{align}\int\left(\frac{\hat{f}'}{\hat{f}}-\frac{f'}{f}\right)^{2}(x)dF(x) & =o_{P}(1)\label{eq:fhat-score-convergence}\\\frac{1}{n}\sum_{i=1}^{n}\left(\frac{\hat{f}'}{\hat{f}}\right)'(X_{i}+\delta_{n}) & =\frac{1}{n}\sum_{i=1}^{n}\left(\frac{f'}{f}\right)'(X_{i})+o_{P}(1)=-I(f)+o_{P}(1).\label{eq:fhat-info-convergence}\end{align}

\begin{remark}
These conditions are satisfied by the kernel density estimator under
mild conditions, see \citep{stone1975adaptive} and \citep{bickel1982adaptive}.
\end{remark}

Let
\begin{align*}
w_{f}(u) & =\frac{1}{I(f)}\left(-\frac{f'}{f}\right)'(F^{-1}(u))
\end{align*}
where the dependency on $F$ is suppressed for notational simplicity.

Let $\hat{f}$ and $\hat{F}$ be the estimates of $f$ and $F$ respectively,
and let
\begin{align*}
w_{\hat{f}}(u) & =\left(\frac{\hat{f}'}{\hat{f}}\right)'(\hat{F}^{-1}(u)) \bigg/ \int_0^1\left(\frac{\hat{f}'}{\hat{f}}\right)'(\hat{F}^{-1}(u))du 
\end{align*}
be the estimate of $w_{f}$. Then $\int_0^1 w_{\hat{f}}(u)du = 1$.

We assume that the following conditions hold for $w_{\hat{f}}$:
\begin{align}
\int_{0}^{1}\frac{\big|w_{\hat{f}}(u)\big|}{f(F^{-1}(u))}du & =O_{P}(\frac{n^{1/4}}{\log n}),\label{eq:wn-rate}
\end{align}
\begin{align}
\int_{0}^{1}\int_{0}^{1}\frac{\min(u,v)-uv}{f(F^{-1}(u))f(F^{-1}(v))}\cdot E\left((w_{\hat{f}}(u)-w_{f}(u))^{2}\right)dudv & =o(1).\label{eq:wn-convergence-condition}
\end{align}

\begin{remark}
Conditions 
\eqref{eq:wn-rate} can be achieved by kernel density estimation with proper tail truncation in the spirit of Lemma 6.1 in \citet{bickel1982adaptive}.
\end{remark}
\begin{remark}
If $X\sim F$ and $var(X)<\infty$, then $\int_{0}^{1}\int_{0}^{1}\frac{\min(u,v)-uv}{f(F^{-1}(u))f(F^{-1}(v))}dudv=var(X)$.
See Bickel (1967). This suggests that the requirement \eqref{eq:wn-convergence-condition}
is not strong.
\end{remark}

\begin{lemma}
\label{lem:linear-order-statistics}Let $X_{1},\cdots,X_{n}$ be i.i.d.
$F$ with $f\equiv F'$ and $\hat{F}$ the empirical distribution.
Let $w_{\hat{f}}$ be an estimate of $w_{f}$ independent of $\{X_{1},\cdots,X_{n}\}$
and satisfy the above conditions 
\eqref{eq:wn-rate} and \eqref{eq:wn-convergence-condition}. Then under the C-R condition,
\begin{align*}
\int_{0}^{1}(\hat{F}^{-1}(u)-F^{-1}(u))\cdot w_{\hat{f}}(u)du & =-\frac{1}{nI(f)}\sum_{i=1}^{n}\frac{f'}{f}(X_{i}) + o_{P}(n^{-1/2}).
\end{align*}
\end{lemma}

\begin{proof}
Note that
\begin{align*}
 & \left|\int_{0}^{1}\left(\hat{F}^{-1}(u)-F^{-1}(u)+\frac{\hat{F}(F^{-1}(u))-u}{f(F^{-1}(u))}\right)\cdot w_{\hat{f}}(u)du\right|\\
\leq & \sup_{0<u<1}\big|f(F^{-1}(u))\cdot(\hat{F}^{-1}(u)-F^{-1}(u))+(\hat{F}(F^{-1}(u))-u)\big|\cdot\int_{0}^{1}\frac{\big|w_{\hat{f}}(u)\big|}{f(F^{-1}(u))}du\\
= & o_{P}(n^{-3/4}\log n)\cdot\int_{0}^{1}\frac{\big|w_{\hat{f}}(u)\big|}{f(F^{-1}(u))}du
\end{align*}
where the last equality is due to Theorem 4 of \citep{Csorgo-Revesz-1978}.
Thus under the further condition \eqref{eq:wn-rate}, we have
\begin{align*}
\int_{0}^{1}\left(\hat{F}^{-1}(u)-F^{-1}(u)+\frac{\hat{F}(F^{-1}(u))-u}{f(F^{-1}(u))}\right)\cdot w_{\hat{f}}(u)du & =o_{P}(n^{-1/2}).
\end{align*}
Note that
\begin{align*}
\int_{0}^{1}\frac{\hat{F}(F^{-1}(u))-u}{f(F^{-1}(u))}w_{f}(u)du & =-\frac{1}{nI(f)}\sum_{i=1}^{n}\frac{f'}{f}(X_{i}).
\end{align*}
Hence it is sufficient to prove that
\begin{align}
\int_{0}^{1}\frac{\hat{F}(F^{-1}(u))-u}{f(F^{-1}(u))}\cdot(w_{\hat{f}}(u)-w_{f}(u))du & =o_{P}(n^{-1/2}).\label{eq:wn-Fhat-o(rootn)}
\end{align}

Recall that $w_{\hat{f}}$ is independent of $\hat{F}$, so
\begin{align*}
 & E\left(\int_{0}^{1}\frac{\hat{F}(F^{-1}(u))-u}{f(F^{-1}(u))}\cdot(w_{\hat{f}}(u)-w_{f}(u))du\right)^{2}\\
= & E\int_{0}^{1}\int_{0}^{1}\frac{\hat{F}(F^{-1}(u))-u}{f(F^{-1}(u))}\cdot\frac{\hat{F}(F^{-1}(v))-v}{f(F^{-1}(v))}\cdot(w_{\hat{f}}-w_{f})(u)\cdot(w_{\hat{f}}-w_{f})(v)dudv\\
= & \frac{1}{n}\int_{0}^{1}\int_{0}^{1}\frac{\min(u,v)-uv}{f(F^{-1}(u))f(F^{-1}(v))}\cdot E\left((w_{\hat{f}}(u)-w_{f}(u))(w_{\hat{f}}(v)-w_{f}(v))\right)dudv\\
\leq & \frac{1}{n}\int_{0}^{1}\int_{0}^{1}\frac{\min(u,v)-uv}{f(F^{-1}(u))f(F^{-1}(v))}\cdot E\left((w_{\hat{f}}(u)-w_{f}(u))^{2}\right)dudv
\end{align*}
which is $o(\frac{1}{n})$ due to \eqref{eq:wn-convergence-condition}.
Therefore \eqref{eq:wn-Fhat-o(rootn)} holds, which concludes the
proof.
\end{proof}

\subsection{Proof of Theorem \ref{thm:efficiency}{\it (i)}}

The median difference, i.e. $\hat{F}_{1}^{_{-1}}(\frac{1}{2})-\hat{F}_{0}^{-1}(\frac{1}{2})$,
is a $\sqrt{n}$ consistent estimate of $\tau$, where $\hat{F}_{1}^{-1}$
and $\hat{F}_{0}^{-1}$ are the empirical quantile functions for $Y$
conditional on $Z=1$ and $Z=0$, respectively. This follows directly
from Theorem A.1.

\subsection{Proof of Theorem \ref{thm:efficiency}{\it (ii)}}

Let $\mathcal{D}_{1}=\{(Z_{i},Y_{i}):1\leq i\leq\frac{n}{2}\}$ and
$\mathcal{D}_{2}=\{(Z_{i},Y_{i}):\frac{n}{2}+1\leq i\leq n\}$ be
a sample split. For simplicity, assume that $\sum_{i=1}^{n/2}Z_{i}=\frac{1}{2}np$
and $\sum_{i=\frac{n}{2}+1}^{n}Z_{i}=\frac{1}{2}np$. Let $\hat{f}_{0(j)}$
be the estimate of $f_{0}\equiv F'_{0}$ based on $\mathcal{D}_{j}$,
$j\in\{1,2\}$, which satisfy the conditions \eqref{eq:fhat-score-convergence}
and \eqref{eq:fhat-info-convergence}.

By the first order Taylor expansion at $\tau$, there exists $\overline{\tau}$
between $\tilde{\tau}$ and $\tau$ such that
\begin{align*}
\hat{\tau}^{if}-\tau\equiv & \tilde{\tau}-\tau+\frac{1}{n}\left(\sum_{i=1}^{n/2}\psi_{\hat{f}_{0(2)}}(Z_{i},Y_{i};\tilde{\tau})+\sum_{i=\frac{n}{2}+1}^{n}\psi_{\hat{f}_{0(1)}}(Z_{i},Y_{i};\tilde{\tau})\right)\\
= & A+(\tilde{\tau}-\tau)(1+B)
\end{align*}
where
\begin{align*}
A & =\frac{1}{n}\left(\sum_{i=1}^{n/2}\psi_{\hat{f}_{0(2)}}(Z_{i},Y_{i};\tau)+\sum_{i=\frac{n}{2}+1}^{n}\psi_{\hat{f}_{0(1)}}(Z_{i},Y_{i};\tau)\right)\\
B & =\frac{1}{n}\cdot\left(\sum_{i=1}^{n/2}\frac{\partial}{\partial\tau}\psi_{\hat{f}_{0(2)}}(Z_{i},Y_{i};\overline{\tau})+\sum_{i=\frac{n}{2}+1}^{n}\frac{\partial}{\partial\tau}\psi_{\hat{f}_{0(1)}}(Z_{i},Y_{i};\overline{\tau})\right).
\end{align*}

Since $\hat{f}_{0(2)}$ is independent of $\{(Z_{i},Y_{i}):1\leq i\leq\frac{n}{2}\}$,
then 
\begin{align*}
 & E\left(\left.\left(\frac{1}{\sqrt{n}}\sum_{i=1}^{n/2}\left(I(\hat{f}_{0(2)})\cdot\psi_{\hat{f}_{0(2)}}(Z_{i},Y_{i};\tau)-I(f_{0})\cdot\psi_{f_{0}}(Z_{i},Y_{i};\tau)\right)\right)^{2}\right|\mathcal{D}_{2}\right)\\
= & \frac{1}{2}E\left(\left.\left(\frac{Z_{1}}{p}\cdot(\frac{\hat{f}_{0(2)}'}{\hat{f}_{0(2)}}-\frac{f_{0}'}{f_{0}})(Y_{1}-\tau)-\frac{1-Z_{1}}{1-p}\cdot(\frac{\hat{f}_{0(2)}'}{\hat{f}_{0(2)}}-\frac{f_{0}'}{f_{0}})(Y_{1})\right)^{2}\right|\mathcal{D}_{2}\right)\\
= & \frac{1}{2p(1-p)}\int\left(\frac{\hat{f}_{0(2)}'}{\hat{f}_{0(2)}}-\frac{f_{0}'}{f_{0}}\right)^{2}(y)dF_{0}(y)
\end{align*}
which is $o_{P}(1)$ according to \eqref{eq:fhat-score-convergence}.
Thus 
\begin{align*}
\frac{I(\hat{f}_{0(2)})}{\sqrt{n}}\cdot\sum_{i=1}^{n/2}\psi_{\hat{f}_{0(2)}}(Z_{i},Y_{i};\tau) & =\frac{I(f_{0})}{\sqrt{n}}\cdot\sum_{i=1}^{n/2}\psi_{f_{0}}(Z_{i},Y_{i};\tau)+o_{P}(1).
\end{align*}
Furthermore, $I(\hat{f}_{0(2)})=I(f_{0})+o_{P}(1)$, hence 
\begin{align*}
\frac{1}{\sqrt{n}}\sum_{i=1}^{n/2}\psi_{\hat{f}_{0(2)}}(Z_{i},Y_{i};\tau) & =\frac{1}{\sqrt{n}}\sum_{i=1}^{n/2}\psi_{f_{0}}(Z_{i},Y_{i};\tau)+o_{P}(1).
\end{align*}
Similarly, we have 
\begin{align*}
\frac{1}{\sqrt{n}}\sum_{i=1+n/2}^{n}\psi_{\hat{f}_{0(1)}}(Z_{i},Y_{i};\tau) & =\frac{1}{\sqrt{n}}\sum_{i=1+n/2}^{n}\psi_{f_{0}}(Z_{i},Y_{i};\tau)+o_{P}(1).
\end{align*}

Therefore, 
\begin{align*}
\sqrt{n}A & =\frac{1}{\sqrt{n}}\sum_{i=1}^{n}\psi_{f_{0}}(Z_{i},Y_{i};\tau)+o_{P}(n^{-1/2})\\
 & \stackrel{d}{\Rightarrow}\mathcal{N}\Bigl(0,\frac{1}{p(1-p)I(f_{0}))}\Bigr).
\end{align*}
From \eqref{eq:fhat-info-convergence}, 
\begin{align*}
B & =\frac{1}{I(f_{0})}\cdot\frac{1}{n}\sum_{i=1}^{n}\frac{Z_{i}}{p}(\frac{f_{0}'}{f_{0}})'(Y_{i}-\tau)+o_{P}(1)\\
 & =-1+o_{P}(1).
\end{align*}

By Slutsky's theorem,
\begin{align*}
\sqrt{n}(\hat{\tau}^{if}-\tau) & =\sqrt{n}A+\sqrt{n}(\tilde{\tau}-\tau)(1+B)\\
 & \stackrel{d}{\Rightarrow}\mathcal{N}\Bigl(0,\frac{1}{p(1-p)I(f_{0})}\Bigr)
\end{align*}
since $\sqrt{n}(\tilde{\tau}-\tau)=O_{P}(1)$ and $1+B=o_{P}(1)$.
This concludes the proof.

\subsection{Proof of Theorem \ref{thm:efficiency}{\it (iii)}}

Let $\mathcal{D}_j$, $j\in\{1,2\}$ be the sample split as above. Let $\hat{F}_{1(j)}$ and $\hat{F}_{0(j)}$
be the empirical distributions of $Y|Z=1$ and $Y|Z=0$ respectively
based on $\mathcal{D}_{j}$, for $j\in\{1,2\}$. Let $w_{\hat{f}_{j}}$
be as defined earlier and assume that $w_{\hat{f}_{j}}$ satisfy the
conditions of Lemma \ref{lem:linear-order-statistics}.

Then the $L$ estimate of $\tau$ can be written as:
\begin{align}
\hat{\tau}^{waq} = & \frac{1}{2}\int_{0}^{1}(\hat{F}_{1(1)}^{-1}(u)-\hat{F}_{0(1)}^{-1}(u))\cdot w_{\hat{f}_{0(2)}}(u)du\nonumber\\  
  & +\frac{1}{2}\int_{0}^{1}(\hat{F}_{1(2)}^{-1}(u)-\hat{F}_{0(2)}^{-1}(u))\cdot w_{\hat{f}_{0(1)}}(u)du.\label{eq:estimator-waq}
\end{align}

According to Lemma \ref{lem:linear-order-statistics}, 
\begin{align*}
\int_{0}^{1}(\hat{F}_{0(1)}^{-1}(u)-F_{0}^{-1}(u))\cdot w_{\hat{f}_{0(2)}}(u)du & =-\frac{2}{n(1-p)I(f_0)}\sum_{i=1}^{n/2}\frac{f_0'}{f_0}(Y_{i})\cdot(1-Z_{i})+o_{P}(n^{-1/2})\\
\int_{0}^{1}(\hat{F}_{1(1)}^{-1}(u)-F_{1}^{-1}(u))\cdot w_{\hat{f}_{0(2)}}(u)du & =-\frac{2}{npI(f_0)}\sum_{i=1}^{n/2}\frac{f_0'}{f_0}(Y_{i}-\tau)\cdot Z_{i}+o_{P}(n^{-1/2})
\end{align*}
where the second equality uses the fact that $F_{1}(\cdot)\equiv F_{0}(\cdot-\tau)$
and $F_{1}^{-1}(u)\equiv F_{0}^{-1}(u)+\tau$. Thus,
\begin{align*}
\int_{0}^{1}(\hat{F}_{1(1)}^{-1}(u)-\hat{F}_{0(1)}^{-1}(u))\cdot w_{\hat{f}_{0(2)}}(u)du= &  -\frac{2}{nI(f_0)}\sum_{i=1}^{n/2}\left(\frac{f_0'}{f_0}(Y_{i}-\tau)\cdot\frac{Z_{i}}{p}-\frac{f_0'}{f_0}(Y_{i})\cdot\frac{1-Z_{i}}{1-p}\right)\\
& +\tau + o_{P}(n^{-1/2}).
\end{align*}
A similar approximation holds for the other term, then by Slutsky's theorem,
\begin{align*}
\hat{\tau}^{waq} & =\tau-\frac{1}{n}\sum_{i=1}^{n}\frac{1}{I(f_0)}\cdot\left(\frac{f_0'}{f_0}(Y_{i}-\tau)\cdot\frac{Z_{i}}{p}-\frac{f_0'}{f_0}(Y_{i})\cdot\frac{1-Z_{i}}{1-p}\right)+o_{P}(n^{-1/2}).
\end{align*}
Hence $\sqrt{n}(\hat{\tau}^{waq}-\tau) \stackrel{d}{\Rightarrow} \mathcal{N}\Bigl(0,\frac{1}{p(1-p)I(f_0)}\Bigr)$.

\section{Proof of partial adaptation with \texorpdfstring{$(\alpha,\beta)$}{(alpha,beta)}-trimmed mean}

\subsection{Main results}
\begin{theorem}
Suppose that $f_0$ falls into the extended Huber family as defined
in Section \ref{section:partial}. If $\alpha$ and $\beta$ are chosen properly, then for the constant additive treatment effect model,
$\hat{\tau}_{\alpha,\beta}$ as defined in (\ref{taualphabeta}) is an efficient estimate of $\tau$.
\end{theorem}

\begin{proof}
According to Theorem 1, the information bound is given by $\frac{1}{p(1-p)}I^{-1}$. 
According to Theorem \ref{thm:quantile-to-mean}, we have
\begin{align*}
    \sqrt{n}(\hat{\tau}_{\alpha,\beta}-\tau) \stackrel{d}{\Rightarrow} \mathcal{N}\Bigl(0, p^{-1}\sigma_{f_1}^2 + (1-p)^{-1}\sigma_{f_0}^2\Bigr), \label{eq:var-atrimmedmean-diff}
\end{align*}
where $\sigma_{f_j}^2$ is defined by \eqref{eq:sigma_f^2} for $j\in \{0, 1\}$, i.e. 
\begin{align}
    \sigma_{f_j}^2 = \frac{1}{(1-\alpha-\beta)^2}\int_{\alpha}^{\beta}\int_{\alpha}^{\beta} \frac{\min(s,t) - st}{f_j(F_j^{-1}(s))f_j(F_j^{-1}(t))}dsdt.
\end{align}
Since $F_{1}^{-1}(t)-F_{0}^{-1}(t)\equiv\tau$, then $\sigma_{f_0}^2 = \sigma_{f_1}^2$. Therefore it is sufficient to prove that $\sigma_{f_0}^2 = I^{-1}$.

For notational simplicity, below we restrict the proof to $\mu=0$ and $\sigma=1$, that is $f_{0}$
is parameterized by $(k_{1},k_{2})$ as \eqref{eq:extended-huber} in Section 2.4 and $f_{1}(x)\equiv f_{0}(x-\tau)$.

Let $\alpha=F_{0}(-k_{1})$ and $1-\beta=F_{0}(k_{2})$.

Then $I=-\int(\frac{f_0'}{f_0})^{'}(x)dF_{0}(x)=\int_{-k_{1}}^{k_{2}}dF_{0}(x)=1-(\alpha+\beta).$

First of all, with integral change of variable,
\begin{align*}
 \sigma_{f_0}^2 = & \frac{1}{(1-\alpha-\beta)^{2}}\int_{-k_{1}}^{k_{2}}\int_{-k_{1}}^{k_{2}}(F_{0}(\min(x,y))-F_{0}(x)F_{0}(y))dxdy.
\end{align*}

Furthermore, it is not hard to verify that $F_{0}(-k_{1})=f_{0}(-k_{1})/k_{1}$
and $1-F_{0}(k_{2})=f_{0}(k_{2})/k_{2}$. Using the fact that $f_{0}'(x)=-xf_{0}(x)$
for $x\in(-k_{1},k_{2})$, then
\begin{align*}
\int_{-k_{1}}^{k_{2}}F_{0}(x)dx & =xF_{0}(x)\big|_{-k_{1}}^{k_{2}}-\int_{-k_{1}}^{k_{2}}xf_{0}(x)dx\\
 & =k_{2}\cdot(1-\beta)+k_{1}\alpha+\int_{-k_{1}}^{k_{2}}df_{0}(x)\\
 & =k_{2}\cdot(1-\beta)+k_{1}\alpha+f_{0}(k_{2})-f_{0}(-k_{1})\\
 & =k_{2}\cdot(1-\beta)+k_{1}\alpha+k_{2}\beta-k_{1}\alpha\\
 & =k_{2}
\end{align*}
and
\begin{align*}
 & \int_{-k_{1}}^{k_{2}}\int_{-k_{1}}^{k_{2}}(F_{0}(\min(x,y))dxdy\\
= & 2\int_{-k_{1}}^{k_{2}}\int_{-k_{1}}^{k_{2}}F_{0}(x)\cdot1(x\leq y)dxdy\\
= & 2\int_{-k_{1}}^{k_{2}}F_{0}(x)\cdot(k_{2}-x)dx\\
= & -\int_{-k_{1}}^{k_{2}}F(x)d(x-k_{2})^{2}\\
= & \int_{-k_{1}}^{k_{2}}(x-k_{2})^{2}f(x)dx+(k_{2}+k_{1})^{2}\alpha\\
= & \int_{-k_{1}}^{k_{2}}\left(k_{2}^{2}f(x)+2k_{2}f'(x)-xf'(x)\right)dx+(k_{2}+k_{1})^{2}\alpha\\
= & k_{2}^{2}(1-\alpha-\beta)+2k_{2}(k_{2}\beta-k_{1}\alpha)+(k_{1}+k_{2})^{2}\alpha-\int_{-k_{1}}^{k_{2}}xf'(x)dx\\
= & k_{2}^{2}+\alpha k_{1}^{2}+\beta k_{2}^{2}-\int_{-k_{1}}^{k_{2}}xdf(x)\\
= & k_{2}^{2}+\int_{-k_{1}}^{k_{2}}f(x)dx\\
= & k_{2}^{2}+1-(\alpha+\beta).
\end{align*}
Thus 
\begin{align*}
\sigma_{f_0}^2 & =\frac{1}{(1-(\alpha+\beta))^{2}}\left((k_{2}^{2}+1-(\alpha+\beta))-k_{2}^{2}\right)\\
               & = I^{-1}.
\end{align*}
which completes the proof.
\end{proof}

\begin{theorem}
Let
\begin{align*}
(\hat{\alpha},\hat{\beta}) & =\arg\min_{\alpha_{0}\leq\alpha,\beta\leq\alpha_{1}}(p^{-1}\hat{\sigma}_{1}^{2}(\alpha,\beta)+(1-p)^{-1}\hat{\sigma}_{0}^{2}(\alpha,\beta))
\end{align*}
where $p^{-1}\hat{\sigma}_{1}^{2}(\alpha,\beta)+(1-p)^{-1}\hat{\sigma}_{0}^{2}(\alpha,\beta)$ is the estimate of the asymptotic variance as derived in \eqref{eq:var-atrimmedmean-diff} and $\hat{\sigma}_j^2$ is the estimate of $\sigma_{f_j}^2$ as given in Proposition \ref{prop:alpha-beta-trim-rate}.
If $0<\alpha_{0},\alpha_{1}<1$ and $\alpha_{0}+\alpha_{1}<1$, and
$\sigma_{1}^{2}(\alpha,\beta)$ has a unique solution w.r.t. $(\alpha,\beta)$,
which falls into $(\alpha_{0},\alpha_{1})\times(\alpha_{0},\alpha_{1})$,
then $(\hat{\alpha},\hat{\beta})$ is a consistent estimate of the
optimal $(\alpha,\beta)$.
\end{theorem}

\begin{proof}
From Proposition \ref{prop:alpha-beta-trim-rate}, uniform convergence below holds
for $j\in\{0,1\}$:
\begin{align*}
\sup_{\alpha_{0}\leq(\alpha,\beta)\leq\alpha_{1}}|\hat{\sigma}_{j}^{2}(\alpha,\beta)-\sigma_{j}^{2}(\alpha,\beta)| & =o_{P}(1).
\end{align*}
The conclusion holds.
\end{proof}

\begin{remark}
By Proposition \ref{lem:convexity}, the limiting objective function $p^{-1}\sigma_{1}^{2}(\alpha,\beta)+(1-p)^{-1}\sigma_{0}^{2}(\alpha,\beta)$
is strictly convex at the local neighborhood of the optimal $(\alpha,\beta)$.
\end{remark}

\subsubsection{Supporting results}
\begin{lemma} \label{lem:var-atrimmedmean}
The influence function \eqref{eq:W-to-psi}, specialized for the $(\alpha,\beta)$-trimmed mean, can be simplified to
\begin{align*}
    \psi(x) = & \frac{1}{1-\alpha-\beta}(x - \mu(\alpha,\beta)), & \text{ if } x \in (F^{-1}(\alpha), F^{-1}(1-\beta)) \\
            = & \frac{1}{1-\alpha-\beta}(F^{-1}(1-\beta) - \mu(\alpha,\beta)), & \text{ if } x \geq F^{-1}(1-\beta) \\
            = & \frac{1}{1-\alpha-\beta}(F^{-1}(\alpha) - \mu(\alpha,\beta)), & \text{ if } x \leq F^{-1}(\alpha)
\end{align*}
where
\begin{align*}
    \mu(\alpha,\beta) & =\int_{\alpha}^{1-\beta}F^{-1}(u)du + \alpha F^{-1}(\alpha) + \beta F^{-1}(1-\beta).
\end{align*}
The asymptotic variance $\sigma_f^2$, denoted as $\sigma^2(\alpha,\beta)$ for the $(\alpha,\beta)$-trimmed mean, can be written as 
\begin{align*}
\sigma^{2}(\alpha,\beta) & = \int_0^1 \psi^2(F^{-1}(u))du. 
\end{align*}
\end{lemma}
\begin{proof}
The proof follows from direct calculation and is omitted.
\end{proof}

\begin{proposition}
\label{prop:alpha-beta-trim-rate}Assume that $f$ is continuous and
that $f(x)>0$ for any $x\in \mathcal{R}$. Let $\hat{\mu}(\alpha, \beta)$ and $\hat{\sigma}^2(\alpha,\beta)$ be the estimate of $\mu(\alpha,\beta)$ and $\hat{\sigma}^2(\alpha,\beta)$ respectively as defined in Lemma \ref{lem:var-atrimmedmean}, by replacing $F$ with $\hat{F}$.
Assume that $\sigma^{2}(\alpha,\beta)$ has a unique minimum w.r.t
$(\alpha,\beta)$ which falls into $(\alpha_{0},\alpha_{1})\times(\alpha_{0},\alpha_{1})$,
where $(\alpha_{0},\alpha_{1})\in(0,1)\times(0,1)$ with $\alpha_{0}+\alpha_{1}<1$.
Let
\begin{align*}
(\hat{\alpha},\hat{\beta}) & =\arg\min_{0<\alpha_{0}\leq\alpha,\beta\leq\alpha_{1}<1}\hat{\sigma}^{2}(\alpha,\beta)
\end{align*}
be the estimate of the optimal $(\alpha,\beta)$. Then $(\hat{\alpha},\hat{\beta})$
is a consistent estimate of the optimal $(\alpha,\beta)$.
\end{proposition}

\begin{proof}
Let $\psi(x)$ be the same as in Lemma \ref{lem:var-atrimmedmean}, and let $\hat{\psi}$ be an estimate of $\psi$  by replacing $F$ with $\hat{F}$. For $u\in (0,1)$, define
\begin{align*}
    e(u) = \psi(F^{-1}(u)) \text{ and }
    \hat{e}(u) =  \hat{\psi}(\hat{F}^{-1}(u)).
\end{align*}

Then
\begin{align*}
|\hat{\sigma}^{2}(\alpha,\beta)-\sigma^{2}| & =|\int_{0}^{1}\hat{e}^{2}(u)du-\int_{0}^{1}e^{2}(u)du|\\
 & =|\int_{0}^{1}(\hat{e}-e)^{2}(u)du+2\int_{0}^{1}e(u)(\hat{e}(u)-e(u))du|\\
 & \leq|\hat{e}-e|_{\infty}^{2}+2|e|_{\infty}\cdot|\hat{e}-e|_{\infty},
\end{align*}
where

\begin{align*}
|e|_{\infty} & =\sup_{\alpha_{0}\leq u\leq1-\alpha_{1}}|e(u)|<\infty\\
|\hat{e}-e|_{\infty} & =\sup_{\alpha_{0}\leq u\leq1-\alpha_{1}}|\hat{e}(u)-e(u)|\\
 & \leq \frac{2}{1-\alpha_0-\alpha_1}\sup_{\alpha_{0}\leq u\leq1-\alpha_{1}}|\hat{F}^{-1}(u)-F^{-1}(u)|\\
 & =O_{P}(n^{-1/2}).
\end{align*}
Therefore the uniform convergence holds:
\begin{align*}
\sup_{0<\alpha_{0}\leq\alpha,\beta\leq\alpha_{1}<1}|\hat{\sigma}^{2}(\alpha,\beta)-\sigma^{2}(\alpha,\beta)| & =o_{P}(1).
\end{align*}
The conclusion follows since the optimal $(\alpha,\beta)$ is unique.
\end{proof}

\begin{proposition}
\label{lem:convexity}Suppose that $f$ is from the extended Huber family as given by \eqref{eq:extended-huber} in Section 2.4. Then $\sigma^2(\alpha,\beta)$ is strictly convex
at the local neighborhood of $(\alpha,\beta)=(F(-k_{1}),1-F(k_{2}))$.
\end{proposition}

\begin{proof}
Note that
\begin{align*}
(1-\alpha-\beta)^{2}\sigma^2(\alpha, \beta) & =\int_{\alpha}^{1-\beta}\int_{\alpha}^{1-\beta}\frac{\min(s,t)-st}{f(F^{-1}(s))f(F^{-1}(t))}dsdt\\
 & \equiv RHS
\end{align*}
then
\begin{align*}
\frac{\partial RHS}{\partial\alpha} & =-2\int_{\alpha}^{1-\beta}\frac{\alpha(1-s)}{f(F^{-1}(s))f(F^{-1}(\alpha))}ds\\
 & =-\frac{2\alpha}{f(F^{-1}(\alpha))}\int_{\alpha}^{1-\beta}\frac{1-s}{f(F^{-1}(s))}ds\\
\end{align*}
and 
\begin{align*}
\frac{\partial RHS}{\partial\beta} & =-\frac{2\beta}{f(F^{-1}(1-\beta))}\int_{\alpha}^{1-\beta}\frac{s}{f(F^{-1}(s))}ds.\\
\end{align*}
Hence
\begin{align*}
\frac{\partial^{2}RHS}{\partial\alpha^{2}} & =\frac{2\alpha}{f(F^{-1}(\alpha))}\frac{1-\alpha}{f(F^{-1}(\alpha))}-2\left(\frac{1-\alpha\frac{f'}{f^{2}}(F^{-1}(\alpha))}{f(F^{-1}(\alpha))}\right)\int_{\alpha}^{1-\beta}\frac{1-s}{f(F^{-1}(s))}ds\\
 & =\frac{2\alpha(1-\alpha)}{f^{2}(F^{-1}(\alpha))}-\frac{2(f-\alpha f'/f)}{f^{2}}(F^{-1}(\alpha))\cdot\int_{\alpha}^{1-\beta}\frac{1-s}{f(F^{-1}(s))}ds
\end{align*}
\begin{align*}
\frac{\partial^{2}RHS}{\partial\beta^{2}} & =\frac{2\beta(1-\beta)}{f^{2}(F^{-1}(1-\beta))}-2\left(\frac{f+\beta f'/f}{f^{2}}\right)(F^{-1}(1-\beta))\cdot\int_{\alpha}^{1-\beta}\frac{s}{f(F^{-1}(s))}ds\\
\frac{\partial RHS}{\partial\alpha\partial\beta} & =\frac{2\alpha\beta}{f(F^{-1}(\alpha))f(F^{-1}(1-\beta))}.
\end{align*}

Now to evaluate the values at
\begin{align*}
\alpha & =F(-k_{1})=f(-k_{1})k_{1}^{-1}\\
\beta & =1-F(k_{2})=f(k_{2})k_{2}^{-1},
\end{align*}
we have 
\begin{align*}
\frac{f'}{f}(F^{-1}(\alpha)) & =k_{1}\\
\frac{f'}{f}(F^{-1}(1-\beta)) & =-k_{2}\\
\alpha k_{1} & =f(-k_{1})\\
\beta k_{2} & =f(k_{2}).
\end{align*}
Thus
\begin{align*}
\int_{\alpha}^{1-\beta}\frac{s}{f(F^{-1}(s))}ds & =\int_{-k_{1}}^{k_{2}}F(x)dx=xF(x)|_{-k_{1}}^{k_{2}}-\int_{-k_{1}}^{k_{2}}xf(x)dx=k_{2}\\
\int_{\alpha}^{1-\beta}\frac{1-s}{f(F^{-1}(s))}ds & =(k_{2}+k_{1})-k_{2}=k_{1}
\end{align*}
and
\begin{align*}
RHS & =1-\alpha-\beta\\
\frac{\partial RHS}{\partial\alpha} & =-2\\
\frac{\partial RHS}{\partial\beta} & =-2
\end{align*}
and
\begin{align*}
\frac{\partial^{2}RHS}{\partial\alpha^{2}} & =\frac{2(1-\alpha)}{k_{1}^{2}\alpha}\\
\frac{\partial^{2}RHS}{\partial\beta^{2}} & =\frac{2(1-\beta)}{k_{2}^{2}\beta}\\
\frac{\partial RHS}{\partial\alpha\partial\beta} & =\frac{2}{k_{1}k_{2}}.
\end{align*}
Hence
\begin{align*}
\frac{\partial\sigma^2(\alpha, \beta)}{\partial\alpha} & =\frac{\partial}{\partial\alpha}\frac{RHS}{(1-\alpha-\beta)^{2}}\\
 & =\frac{-2}{(1-\alpha-\beta)^{2}}+RHS\cdot\frac{2}{(1-\alpha-\beta)^{3}}=0\\
\frac{\partial\sigma^2(\alpha, \beta)}{\partial\beta} & =\frac{\partial}{\partial\beta}\frac{RHS}{(1-\alpha-\beta)^{2}}=\frac{-2}{(1-\alpha-\beta)^{2}}+RHS\cdot\frac{2}{(1-\alpha-\beta)^{3}}=0
\end{align*}
and the Hessian matrix can be calculated as below:
\begin{align*}
2\cdot(1,1)^{T}(1,1)\cdot\sigma^2(\alpha,\beta)+(1-\alpha-\beta)^{2}\frac{\partial^{2}\sigma^2(\alpha, \beta)}{\partial(\alpha,\beta)\partial(\alpha,\beta)} & =\left(\begin{array}{cc}
\frac{2(1-\alpha)}{k_{1}^{2}\alpha} & \frac{2}{k_{1}k_{2}}\\
\frac{2}{k_{1}k_{2}} & \frac{2(1-\beta)}{k_{2}^{2}\beta}
\end{array}\right)\\
\end{align*}
so
\begin{align*}
\frac{\partial^{2}\sigma^2(\alpha, \beta)}{\partial(\alpha,\beta)\partial(\alpha,\beta)} & =\frac{2}{(1-\alpha-\beta)^{2}}\left(\begin{array}{cc}
\frac{(1-\alpha)}{k_{1}^{2}\alpha}-\frac{1}{1-\alpha-\beta} & \frac{1}{k_{1}k_{2}}-\frac{1}{1-\alpha-\beta}\\
\frac{1}{k_{1}k_{2}}-\frac{1}{1-\alpha-\beta} & \frac{(1-\beta)}{k_{2}^{2}\beta}-\frac{1}{1-\alpha-\beta}
\end{array}\right).\\
\end{align*}
Note that 
\begin{align*}
\int_{-k_{1}}^{k_{2}}x^{2}f(x)dx & =\int_{-k_{1}}^{k_{2}}-xf'(x)dx\\
 & =\int_{-k_{1}}^{k_{2}}f(x)dx-xf(x)|_{-k_{1}}^{k_{2}}\\
 & =1-\alpha-\beta-(k_{2}^{2}\beta+k_{1}^{2}\alpha)
\end{align*}
that is, 
\begin{align*}
1-\alpha-\beta & =k_{1}^{2}\alpha+k_{2}^{2}\beta+\int_{-k_{1}}^{k_{2}}x^{2}f(x)dx\\
 & =k_{1}^{2}\alpha+k_{2}^{2}\beta+(1-\alpha-\beta)ES^{2}
\end{align*}
where $S$ is defined on $x\in(-k_{1},k_{2})$ with density $\frac{f(x)}{1-\alpha-\beta}$.
Then 
\begin{align*}
ES & =(1-\alpha-\beta)^{-1}\int_{-k_{1}}^{k_{2}}xf(x)dx\\
 & =(1-\alpha-\beta)^{-1}\int_{-k_{1}}^{k_{2}}-f'(x)dx\\
 & =\frac{k_{1}\alpha-k_{2}\beta}{1-\alpha-\beta}
\end{align*}

So the determinant of $\frac{\partial^{2}\sigma^2(\alpha,\beta)}{\partial(\alpha,\beta)\partial(\alpha,\beta)}\times\frac{(1-\alpha-\beta)^{2}}{2}$
is
\begin{align*}
 & (\frac{(1-\alpha)}{k_{1}^{2}\alpha}-\frac{1}{1-\alpha-\beta})(\frac{(1-\beta)}{k_{2}^{2}\beta}-\frac{1}{1-\alpha-\beta})-(\frac{1}{k_{1}k_{2}}-\frac{1}{1-\alpha-\beta})^{2}\\
= & \frac{(1-\alpha)}{k_{1}^{2}\alpha}\frac{(1-\beta)}{k_{2}^{2}\beta}-\frac{1}{k_{1}^{2}k_{2}^{2}}+\frac{1}{1-\alpha-\beta}(\frac{2}{k_{1}k_{2}}-\frac{(1-\alpha)}{k_{1}^{2}\alpha}-\frac{(1-\beta)}{k_{2}^{2}\beta})\\
= & \frac{k_{1}^{2}\alpha+k_{2}^{2}\beta+(1-\alpha-\beta)ES^{2}}{k_{1}^{2}k_{2}^{2}\alpha\beta}-\frac{1}{1-\alpha-\beta}\cdot\frac{k_{1}^{2}\alpha(1-\beta)+k_{2}^{2}\beta(1-\alpha)-2k_{1}k_{2}\alpha\beta}{k_{1}^{2}k_{2}^{2}\alpha\beta}\\
= & \frac{(1-\alpha-\beta)^{2}ES^{2}-(k_{1}^{2}\alpha^{2}+k_{2}^{2}\beta^{2}-2k_{1}k_{2}\alpha\beta)}{(1-\alpha-\beta)k_{1}^{2}k_{2}^{2}\alpha\beta}\\
= & \frac{1-\alpha-\beta}{k_{1}^{2}k_{2}^{2}\alpha\beta}\cdot var(S)
\end{align*}
which is greater than 0 since $\alpha+\beta<1$. This completes the
proof.
\end{proof}

\section{Proof of Theorem \ref{thm:eif-est-general}}

\subsection{Proof of Lemma \ref{thm:eif-general}}

Let $g(y,\theta)=\frac{\partial}{\partial\theta}\log f_{1}(y,\theta)$,
where $f_{1}(y,\theta)$ is the same as in \eqref{eq:f1(y,theta)}.

In order to find the projection of $\dot{\ell}$ on the orthocomplement
of the $\dot{P}_{f}$, we need to find $\text{Q}$ which minimizes
\begin{align*}
\epsilon(Q)\equiv & E\left(\dot{\ell}(Y,Z,\theta)-(ZQ(h^{-1}(Y,\theta),\theta)+(1-Z)Q(Y,\theta))\right)^{2}.
\end{align*}

Since $Z\in\{0,1\}$ and $\dot{\ell}(y,z,\theta)=z\cdot g(y,\theta)$,
we have 
\begin{align*}
\epsilon(Q) & =E\left(Z\cdot(g(Y,\theta)-Q(h^{-1}(Y,\theta),\theta))^{2}\right)+E\left((1-Z)\cdot Q^{2}(Y,\theta)\right)\\
 & =p\cdot E\left((g(Y,\theta)-Q(h^{-1}(Y,\theta),\theta))^{2}|Z=1\right)+(1-p)\cdot E\left(Q^{2}(Y,\theta)|Z=0\right)\\
 & =p\cdot E\left((g(h(Y,\theta),\theta)-Q(Y,\theta))^{2}|Z=0\right)+(1-p)\cdot E\left(Q^{2}(Y,\theta)|Z=0\right)
\end{align*}
where the last equality is due to $P(Y\leq h(y,\theta)|Z=1)=P(Y\leq y|Z=0)$.
Therefore,
\begin{align*}
\epsilon(Q) & =E\left((Q(Y,\theta)-p\cdot g(h(Y,\theta),\theta))^{2}|Z=0\right)+const
\end{align*}
where $const$ does not depend on $Q$. This leads to
\begin{align*}
Q(y,\theta) & =p\cdot g(h(y,\theta),\theta).
\end{align*}
Note that this $Q$ satisfies $\int Q(y,\theta)f_{0}(y)dy=0$ since
\begin{align*}
    E\frac{\partial\ell}{\partial\theta}(Y,Z,\theta) &= E(Zg(Y,\theta))\\
    &= p\cdot E(g(Y,\theta)|Z=1) \\
    &= p\cdot E(g(h(Y,\theta),\theta)|Z=0)
\end{align*}
must be 0.
Thus
\begin{align*}
\dot{\ell}^*(y,z,\theta) & =z\cdot g(y,\theta)-(z\cdot p\cdot g(y,\theta)+(1-z)\cdot p\cdot g(h(y,\theta),\theta))\\
 & =(1-p)\cdot z\cdot g(y,\theta)-p\cdot(1-z)\cdot g(h(y,\theta),\theta).
\end{align*}
and
\begin{align*}
 E(\dot{\ell}^*(Y,Z,\theta))^{2} & =(1-p)^{2}E(Z\cdot g^{2}(Y,\theta))+p^{2}E((1-Z)\cdot g^{2}(h(Y,\theta),\theta))\\
 & =p(1-p)\int g^{2}(h(y,\theta),\theta)f_{0}(y)dy.
\end{align*}

\subsection{Proof of Theorem \ref{thm:eif-est-general}}

The proof is similar to the sample-splitting argument for Claim 3 of Theorem 1 and thus is omitted for conciseness.

\section{Details on the Empirical Implementation}\label{sec:empirical-implementation}

\subsection{Density Estimation}

Our estimators rely on estimates of the density and its first and second derivative.

Using data \(X\) (for instance half the sample of control observations for the estimator with sample splitting), we estimate \(f^{(k)}(x)\) at \(x \in \mathbb{X} \equiv \{X_{(\lceil n_{d} \frac{k}{1000} \rceil)} \; | \; k=1,\dots,999\}\) where \(n_{d}\) is the number of observations in the data \(X\) used for estimating densities.
Estimating the density only at points in the density-estimation sample avoids estimating densities of \(0\) in the tails where there are few observations especially when using sample-splitting.
We use adaptive kernel estimates (Silverman 1986, ch. 5.3) and a triweight kernel, which allows estimation of the first and second derivative of the density. An alternative based on direct estimation of the logarithm of the density and its derivatives is \citet{pinkse2021estimates}.

Density estimation proceeds in two steps.
First, we obtain an initial estimate of the density at each point \(x \in \mathbb{X}\) as
\begin{equation*}
\tilde{f}(x) = \frac{1}{n_{d} h} \sum_{i=1}^{n_{d}} K(\frac{x - X_{i}}{h})
\end{equation*}
where \(K(u)\) is the triweight kernel
\begin{equation*}
K(u) = \frac{35}{32}(1-u^{2})^{3} \cdot 1_{|u| \leq 1}
\end{equation*}
and \(h\) is the bandwidth, chosen as
\begin{equation*}
h = 3.15 \cdot \hat{\sigma}_{X} \cdot n_{d}^{-1/5}
\end{equation*}
where \(3.15\) is the Silverman rule-of-thumb constant for the triweight kernel and
\begin{equation*}
\hat{\sigma}_{X} = \frac{\text{quantile}(X,0.95) - \text{quantile}(X,0.05)}{2 \cdot 1.6449}
\end{equation*}
is an estimate of the standard deviation of \(X\) if \(X\) is normally distributed but less affected by outliers than the conventional unbiased estimator of the variance.

Second, we obtain our final estimate of the density by calculating an \emph{adaptive} bandwidth.
Intuitively, in regions of high density, we can choose a small bandwidth that lowers bias without incurring much of a cost in terms of variance.
In regions of low density, however, a larger bandwidth is needed to have sufficiently many data points near \(x\) to reign in the variance.
Following Silverman (1986, ch. 5.3.1), define the local bandwidth factors \(\lambda_{i}\) for \(X_{i}\) as
\begin{equation*}
\lambda_{i} = \Bigl(\frac{\tilde{f}(X_{i})}{g} \Bigr)^{-\alpha}
\end{equation*}
where \(g\) is the geometric mean of the \(\tilde{f}(X_{i})\):
\begin{equation*}
g = \Bigl(\prod_{i=1}^{n_{d}} \tilde{f}(X_{i})\Bigr)^{1/n_{d}}.
\end{equation*}
Then the estimate of the density is
\begin{equation*}
\hat{f}(x) = \frac{1}{n_{d}} \sum_{i=1}^{n_{d}} \frac{1}{h \lambda_{i}} K(\frac{x - X_{i}}{h \lambda_{i}}).
\end{equation*}

To estimate the first and second derivative of the density, we take \(\lambda_{i}\) as above, and estimate
\begin{equation*}
\begin{aligned}
\hat{f}'(x) & = \frac{1}{n_{d}} \sum_{i=1}^{n_{d}} \frac{1}{(h_{1} \lambda_{i})^{2}} K'(\frac{x - X_{i}}{h_{1} \lambda_{i}}) \\
\hat{f}''(x) & = \frac{1}{n_{d}} \sum_{i=1}^{n_{d}} \frac{1}{(h_{2} \lambda_{i})^{3}} K''(\frac{x - X_{i}}{h_{2} \lambda_{i}})
\end{aligned}
\end{equation*}
with \(h_{1}\) and \(h_{2}\) the Silverman rule-of-thumb bandwidth for estimating the derivatives of the density using a triweight kernel
\begin{equation*}
\begin{aligned}
h_{1} & = 2.83 \cdot \hat{\sigma}_{X} \cdot n_{d}^{-1/7} \\
h_{2} & = 2.70 \cdot \hat{\sigma}_{X} \cdot n_{d}^{-1/9} \\
\end{aligned}
\end{equation*}
and derivatives of the triweight kernel \(K\)
\begin{equation*}
\begin{aligned}
K'(u)  &= \frac{3 \cdot 35}{32}(1-u^{2})^{2} \cdot (-2 u) \cdot 1_{|u| \leq 1} \\
K''(u) &= \frac{2 \cdot 3 \cdot 35}{32}(1-u^{2}) \cdot (5 u^{2} - 1) \cdot 1_{|u| \leq 1} \\
\end{aligned}
\end{equation*}

Then our estimates of the derivatives of the \emph{log} density are, for \(x \in \mathbb{X}\),
\begin{equation*}
\begin{aligned}
\widehat{\frac{\partial \ln f(x)}{\partial x}} & = \frac{\hat{f}'(x)}{\hat{f}(x)} \\
\widehat{\frac{\partial^{2} \ln f(x)}{\partial x^{2}}} & = \frac{\hat{f}(x) \hat{f}''(x) - (\hat{f}'(x))^{2}}{(\hat{f}(x))^{2}} \\
\end{aligned}
\end{equation*}
When we need to evaluate the second derivative at some other point \(x \not\in \mathbb{X}\), we use linear interpolation and nearest-neighbor extrapolation of the estimated derivative of the log density at the points in \(\mathbb{X}\).
In practice, this piecewise-linear approximation causes negligible error because areas where \(\mathbb{X}\) is not tightly spaced have low density by construction and typically very little curvature.

\subsection{Efficient Influence Function Estimator}

The M-estimator given in equation \ref{eq:tau-if} of the main text, substituting equation \ref{eq:psi-additive} for \(\psi\) and without sample splitting for notational convenience, is:
\begin{equation*}
\hat{\tau}^{if} = \tilde{\tau}+\frac{1}{n}\sum_{i=1}^{n}
-\frac{1}{I(f_0)}\cdot\left(\frac{Z}{p}\cdot\frac{f_0'}{f_0}(Y-\tilde{\tau})-\frac{1-Z}{1-p}\cdot\frac{f_0'}{f_0}(Y)\right)
\end{equation*}
The densities can be estimated as described in the previous section.
Rather than this one-step estimate using a single initial estimate \(\tilde{\tau}\) (for instance the difference in medians), we implement the estimator that is the fixed point of the equation above; that is, \(\hat{\tau}^{if} = \tilde{\tau}\).
At this fixed point, the equation simplifies, such that the efficient influence function estimator is the solution to the equation:
\begin{equation}\label{eq:impl-theta-eif-0}
\frac{1}{n_{1}} \sum_{i=1}^{n_{1}}
    \widehat{\frac{\partial \ln f}{\partial x}} (Y_{i}^{(1)} - \hat{\theta}^{if})
-
\frac{1}{n_{0}} \sum_{i=1}^{n_{0}} 
    \widehat{\frac{\partial \ln f}{\partial x}} (Y_{i}^{(0)})
= 0
\end{equation}
where the estimates \(\widehat{\frac{\partial \ln f}{\partial x}}\) either come from the full sample or we use sample-splitting as described in Section 2.
Standard root-finding algorithms appear to work well for finding the solution \(\hat{\theta}^{if}\) to equation~\ref{eq:impl-theta-eif-0} when initialized at reasonable guesses such as the weighted average quantiles estimator (see below) or the difference in medians.

We estimate standard errors as follows:
First, calculate
\begin{equation*}
\hat{I} = -\frac{1}{n_{0}} \sum_{i=1}^{n_{0}} \widehat{\frac{\partial^{2} \ln f}{\partial x^{2}}} (Y_{i}^{(0)})
\end{equation*}
where \(Y_{i}^{(0)}\) for \(i = 0,\dots,n_{0}\) are the control observations.
Then an estimate of the approximate finite sample variance of \(\hat{\theta}^{if}\) is
\begin{equation*}
\widehat{\var(\hat{\theta}^{if})} = \frac{1}{p(1-p)} \frac{1}{\hat{I}} \frac{1}{n}
\end{equation*}
where as before \(p = \frac{n_{1}}{n_{0}+n_{1}}\) is the fraction of treated observations.
This estimator of the variance relies on the correct specification of the additive shift model.

\subsection{Weighted Average Quantile Estimator}

For the weighted average quantiles estimator, we calculate, assuming $n_0 \geq n_1$
\begin{equation}\label{eq:impl-theta-waq}
\hat{\theta}^{waq}
=\sum_{i=1}^{n_{0}} w\bigl(\frac{i}{n_{0}+1}\bigr) \bigl(Y_{(\lceil n_{1}\frac{i}{n_{0}+1} \rceil)}^{(1)} - Y_{(i)}^{(0)}\bigr)
\end{equation}
where for $u\in \{i/(n_0+1): i=1,\cdots,n_0\}$
\begin{equation*}
\begin{aligned}
\tilde{w}(u) & = \widehat{\frac{\partial^{2} \ln f}{\partial x^{2}}} \bigl(\text{quantile}(Y^{(0)},u)\bigr) \\
w(u) & =
\begin{cases}
c \cdot \tilde{w}(u) 
    & \text{if } \left| \frac{\tilde{w}(u)/\sum_{i=1}^{n_{0}} \tilde{w}\bigl(\frac{i}{n_{0}+1}\bigr)}
    {\hat{f}(\hat{F}^{-1}(u)) \cdot \text{MAD}(\hat{f})} \right| 
        < \log(\log(n)) / \log(n) \cdot n^{1/4} \\
0 & \text{otherwise}.
\end{cases}
\end{aligned}
\end{equation*}
Here $c$ is for normalization such that $\sum_{i=1}^{n_0}w(\frac{i}{n_0+1})=1$, and \(\text{quantile}(Y^{(0)},u)\) is the \(u\) quantile of all the control observations for the full-sample estimator.
\(\hat{f}(\hat{F}^{-1}(u))\) is an ad-hoc estimate of the density at the \(u\) quantile calculated as \(\hat{f}(\hat{F}^{-1}(i/(n_{0}+1))) = (1/n_{0}) / \max(\Delta Y_{(i)}^{(0)},\Delta Y_{(\lceil n_{1}\frac{i}{n_{0}+1} \rceil)}^{(1)})\) with \(\Delta\) the first difference operator. The factor $\text{MAD}(\hat{f})$, i.e. the median absolute deviation of the data (multiplied by $1.4826$ such that it estimates the standard deviation if the data are normally distributed), is used so that the truncation condition is both shift-invariant and scale-invariant. 
The truncation of the weights \(\tilde{w}\) addresses violations of condition \eqref{eq:wn-rate} in the tails by truncating observations with excessive (normalized) weight relative to the density at the quantile.
In the simulations, it improves the performance for the Cauchy distribution for some samples with extreme outliers, and otherwise has no effect.

For a sample-splitting estimator, we randomly split the samples of treated and control in half, use the first half to estimate \(w(u)\) (the densities and \(\text{quantile}(Y^{(0)},u)\)), and then evaluate equation \ref{eq:impl-theta-waq} for the second half of the sample given \(w\).
Reverse the roles of the sample splits and average the two estimators, as described in the main text.

Under correct specification, \(\hat{\theta}^{waq}\) and \(\hat{\theta}^{if}\) are asymptotically equivalent, so the standard errors based on \(\hat{I}\) are appropriate for \(\hat{\theta}^{waq}\) as well.

\section{Results With Sample Splitting}

Table \ref{tab:sample-splitting} reports results for the influence function-based and weighted average quantile estimators with and without sample splitting for the simulations shown in Tables~\ref{tab:simulations-parametric}, \ref{tab:house-prices-simulations}, and \ref{tab:health-simulations-log}.

\begin{table}
    \caption{\label{tab:sample-splitting}Results for the simulations with and without sample splitting.}
    \centering
\resizebox{1\linewidth}{!}{
\begin{tabular}{l rrrrr rr rr}
\toprule
& & & & & & \multicolumn{2}{c}{95\% C.I.} & \multicolumn{2}{c}{boot. var. C.I.} \\
\cmidrule(lr){7-8} \cmidrule(lr){9-10}
estimator & bias & \makecell{standard \\ deviation} & \makecell{relative \\ efficiency} & RMSE & MAD & coverage & \makecell{median \\ length} & coverage & \makecell{median \\ length}  \\
\midrule
\multicolumn{10}{l}{\emph{Parametric Distributions}}\\
\multicolumn{10}{l}{Normal Distribution:}\\
eif: full sample 
& 0.000 & 0.014 & 1.02 & 0.014 & 0.010 & 0.95 & 0.056 & 0.95 & 0.057 \\ 
eif: split sample & 0.000 & 0.014 & 1.02 & 0.014 & 0.010 & 0.95 & 0.056 & 0.95 & 0.057 \\ 
waq: full sample 
& 0.000 & 0.014 & 1.02 & 0.014 & 0.010 & 0.95 & 0.056 & 0.95 & 0.056 \\ 
waq: split sample & 0.000 & 0.014 & 1.02 & 0.014 & 0.010 & 0.95 & 0.056 & 0.95 & 0.057 \\ 
\multicolumn{10}{l}{Double Exponential Distribution:}\\
eif: full sample 
& -0.000 & 0.015 & 1.06 & 0.015 & 0.010 & 0.95 & 0.060 & 0.96 & 0.060 \\ 
eif: split sample & -0.000 & 0.015 & 1.07 & 0.015 & 0.010 & 0.95 & 0.060 & 0.96 & 0.061 \\ 
waq: full sample 
& -0.000 & 0.015 & 1.07 & 0.015 & 0.010 & 0.95 & 0.060 & 0.96 & 0.061 \\ 
waq: split sample & -0.000 & 0.015 & 1.08 & 0.015 & 0.010 & 0.95 & 0.060 & 0.96 & 0.062 \\ 
\multicolumn{10}{l}{Cauchy Distribution:}\\
eif: full sample 
& 0.000 & 0.020 & 1.01 & 0.020 & 0.014 & 0.97 & 0.085 & 0.97 & 0.088 \\ 
eif: split sample & 0.000 & 0.021 & 1.03 & 0.021 & 0.014 & 0.96 & 0.085 & 0.98 & 0.093 \\ 
waq: full sample 
& 0.000 & 0.021 & 1.05 & 0.021 & 0.014 & 0.96 & 0.085 & 0.98 & 0.099 \\ 
waq: split sample & 0.000 & 0.022 & 1.08 & 0.022 & 0.015 & 0.95 & 0.085 & 0.99 & 0.110 \\ 
\\
\multicolumn{10}{l}{\emph{House Price Data}}\\
\multicolumn{10}{l}{effect in levels based on additive model in levels}\\
eif: full sample 
& 13 & 941 & 0.97 & 941 & 628 & 0.95 & 3819 & 0.95 & 3768 \\
eif: split sample & 8 & 944 & 0.97 & 944 & 627 & 0.95 & 3819 & 0.95 & 3830 \\ 
waq: full sample 
& 13 & 946 & 0.97 & 946 & 626 & 0.95 & 3819 & 0.95 & 3859 \\
waq: split sample & 7 & 953 & 0.98 & 953 & 628 & 0.95 & 3819 & 0.96 & 3935 \\ 
\multicolumn{10}{l}{multiplicative parameter: additive model in logs}\\
eif: full sample 
& -0.0000 & 0.0067 & 1.08 & 0.0067 & 0.0045 & 0.95 & 0.026 & 0.95 & 0.026 \\ 
eif: split sample & -0.0000 & 0.0067 & 1.09 & 0.0067 & 0.0045 & 0.95 & 0.026 & 0.95 & 0.027 \\ 
waq: full sample 
& -0.0000 & 0.0067 & 1.08 & 0.0067 & 0.0044 & 0.95 & 0.026 & 0.95 & 0.027 \\ 
waq: split sample & -0.0000 & 0.0067 & 1.09 & 0.0067 & 0.0044 & 0.95 & 0.026 & 0.96 & 0.027 \\ 
\multicolumn{10}{l}{effect in levels based on additive model in logs}\\
eif: full sample & -5 & 1363 & 1.08 & 1363 & 914 & 0.95 & 5374 & 0.95 & 5415 \\ 
eif: split sample & -2 & 1367 & 1.09 & 1367 & 913 & 0.95 & 5374 & 0.95 & 5480 \\ 
waq: full sample & -3 & 1364 & 1.08 & 1364 & 910 & 0.95 & 5374 & 0.95 & 5462 \\ 
waq: split sample & -1 & 1369 & 1.09 & 1369 & 909 & 0.95 & 5374 & 0.96 & 5556 \\ 

\multicolumn{10}{l}{Medical Expenditures Data}\\
\multicolumn{10}{l}{effect in levels based on additive model in levels}\\
eif: full sample 
& 3928 & 682 & -- & 3987 & 3922 & 0.01 & 4878 & 0.00 & 3240 \\
eif: split sample & 3884 & 743 & -- & 3954 & 3884 & 0.02 & 4878 & 0.00 & 3670 \\
waq: full sample 
& 3846 & 792 & -- & 3927 & 3848 & 0.03 & 4878 & 0.01 & 4018 \\
waq: split sample & 3830 & 880 & -- & 3930 & 3820 & 0.05 & 4878 & 0.05 & 4713 \\
\multicolumn{10}{l}{multiplicative parameter: additive model in logs}\\
eif: full sample 
& 0.029 & 0.066 & -- & 0.072 & 0.049 & 0.95 & 0.27 & 0.94 & 0.26 \\
eif: split sample & 0.030 & 0.066 & -- & 0.073 & 0.048 & 0.95 & 0.27 & 0.94 & 0.27 \\
waq: full sample 
& 0.028 & 0.066 & -- & 0.071 & 0.047 & 0.95 & 0.27 & 0.95 & 0.27 \\
waq: split sample & 0.027 & 0.066 & -- & 0.072 & 0.048 & 0.95 & 0.27 & 0.95 & 0.27 \\
\multicolumn{10}{l}{effect in levels based on additive model in logs}\\
eif: full sample 
& 3532 & 2525 & -- & 4341 & 3481 & 0.78 & 10452 & 0.76 & 10014 \\ 
eif: split sample & 3565 & 2545 & -- & 4380 & 3511 & 0.78 & 10465 & 0.77 & 10149 \\ 
waq: full sample 
& 3464 & 2523 & -- & 4285 & 3413 & 0.79 & 10441 & 0.77 & 10147 \\ 
waq: split sample & 3437 & 2538 & -- & 4272 & 3384 & 0.79 & 10433 & 0.80 & 10323 \\
\bottomrule
\end{tabular}
}
\end{table}

\end{document}